\documentclass[journal]{IEEEtran}
 \usepackage{amsmath,amssymb}
 \usepackage{subfigure}
 \usepackage{graphicx,graphics,color,psfrag}
 \usepackage{cite,balance}
 \usepackage{caption}
 \allowdisplaybreaks
 \usepackage{algorithm}
 \usepackage{accents}
 \usepackage{amsthm}
 \usepackage{bm}
 \usepackage{algorithmic}
 \usepackage[english]{babel}
 \usepackage{multirow}
 \usepackage{enumerate}
 \usepackage{cases}
 \usepackage{stfloats}
 \usepackage{dsfont}
 \usepackage{color,soul}
 \usepackage{amsfonts}
 \usepackage{cite,graphicx,amsmath,amssymb}
 \usepackage{subfigure}
 \usepackage{fancyhdr}
 \usepackage{hhline}
 \usepackage{graphicx,graphics}
 \usepackage{array,color}
 \usepackage{amsmath}

\newtheorem{theorem}{Theorem}

\newtheorem{lemma}{Lemma}
\newtheorem{corollary}{Corollary}

\newtheorem{proposition}{Proposition}

\newtheorem{remark}{\bf Remark}
\def\phi{\varphi}

\def\l{\left}
\def\r{\right}
\def\({\left(}
\def\){\right)}

\def\b0{{\mathbf{0}}}

\newcommand{\nn}{\nonumber}

\begin{document}

\title{\huge Energy-Efficient  Resource Allocation for \\ Mobile-Edge Computation Offloading} 
\author{Changsheng You, Kaibin Huang, Hyukjin Chae and Byoung-Hoon Kim   \thanks{\noindent C. You and K. Huang are  with the Dept. of EEE at The  University of  Hong Kong, Hong Kong (Email: csyou@eee.hku.hk, huangkb@eee.hku.hk). H. Chae  and B.-H. Kim are with LG Electronics,  S. Korea (Email: hyukjin.chae@lge.com, bh.kim@lge.com). This work was supported by a grant from LG Electronics. Part of this work has been submitted to IEEE Globecom 2016.}}
\maketitle

\begin{abstract} 
\emph{Mobile-edge computation offloading} (MECO) offloads intensive mobile computation to clouds located at the edges of cellular networks. Thereby, MECO is envisioned as a  promising technique for  prolonging  the battery lives and enhancing the computation capacities of mobiles.  In this paper, we study resource allocation for a multiuser MECO system based on time-division multiple access (TDMA) and orthogonal frequency-division multiple access (OFDMA). First, for the TDMA MECO system with infinite or finite cloud computation capacity, the optimal resource allocation is formulated as a convex optimization problem  for minimizing  the weighted sum   mobile energy consumption under the constraint on computation latency. The optimal policy is proved to have a threshold-based structure with respect to a derived \emph{offloading priority function}, which  yields priorities  for users according  to their channel gains and local computing energy consumption.  As a result, users with priorities above and below  a given  threshold perform complete  and minimum offloading, respectively. Moreover, for the cloud with finite capacity, a sub-optimal resource-allocation algorithm  is proposed to reduce the computation complexity for computing the threshold. Next, we consider the OFDMA MECO system, for which the optimal resource allocation is formulated as a  mixed-integer problem. To solve this challenging problem and characterize its policy structure, a low-complexity sub-optimal  algorithm is proposed by transforming the OFDMA problem to its TDMA counterpart. The corresponding resource allocation  is derived by defining an \emph{average offloading priority function} and shown to have close-to-optimal performance in simulation. 
\end{abstract}
\begin{IEEEkeywords}
Mobile-edge computing, resource allocation, mobile computation offloading, energy-efficient computing.
\end{IEEEkeywords}

\section{Introduction}
The realization of Internet of Things (IoT) \cite{Melanie:Sensor:2012} will connect tens of billions of resource-limited  mobiles, e.g., mobile devices, sensors and wearable computing devices, to Internet via cellular networks. The finite battery lives and limited computation capacities of mobiles pose  significant challenges for designing  IoT. One promising solution is to leverage mobile-edge computing \cite{patel2014mobile} and offload intensive mobile computation to nearby clouds at the edges of cellular networks, called \emph{edge clouds},  with short latency, referred to as \emph{mobile-edge computation offloading} (MECO).  In this paper, we consider a MECO system with a single edge cloud serving multiple users and investigate the energy-efficient resource allocation.  

\subsection{Prior Work}
Mobile computation offloading (MCO) \cite{kumar:CanOffloadSave2010}  (or mobile cloud computing) has been extensively studied in  computer science, including system architectures (e.g., MAUI \cite{cuervo:MAUI:2010}), virtual machine migration \cite{xiao2013dynamic} and power management \cite{van2010performance}.  It is commonly assumed that the implementation of MCO relies on a network architecture with a central cloud (e.g., a data center). This architecture has the drawbacks of high overhead  and long backhaul latency \cite{Ahmed2015ahmedsurvey}, and will soon encounter the performance bottleneck of finite backhaul capacity in view of exponential mobile traffic  growth. These issues can be overcome by MECO based on a network architecture supporting distributed mobile-edge computing. Among others, designing energy-efficient control policies is a key challenge for the MECO system.

Energy-efficient MECO requires the joint design of MCO and wireless communication techniques. Recent years have seen research progress on this topic for both single-user \cite{zhang:MobileMmodel:2013, you2015energyJSAC,dynamic2016zhang,xiang2014energy} and multiuser \cite{chen2015efficient,Barbarossa:MobileCloud:2014,zhao2015cooperative,ge2012game,kaewpuang2013framework } MECO systems.  For a single-user MECO system, the optimal offloading decision policy was derived in \cite{zhang:MobileMmodel:2013} by comparing the energy consumption of optimized local computing (with variable CPU cycles) and offloading (with  variable transmission rates). This framework was  further developed in \cite{you2015energyJSAC} and \cite{dynamic2016zhang} to enable adaptive offloading powered by wireless  energy transfer and energy harvesting, respectively. Moreover,  dynamic offloading was integrated with adaptive LTE/WiFi link selection in \cite{xiang2014energy} to achieve higher energy efficiency.  For  multiuser MECO systems, the control policies for energy savings are more complicated. In \cite{chen2015efficient}, distributed computation  offloading for multiuser MECO at a single cloud was designed using  game theory for both energy-and-latency minimization at mobiles.  A multi-cell MECO system was considered in \cite{Barbarossa:MobileCloud:2014}, where the radio and computation resources were jointly allocated to minimize the mobile energy consumption under offloading latency  constraints.  With the coexistence of central and edge clouds, the optimal user scheduling for offloading to different clouds was studied in \cite{zhao2015cooperative}. In addition to total mobile energy consumption,  cloud energy consumption for computation was also minimized  in \cite{ge2012game} by designing the mapping between clouds and mobiles for offloading  using game theory.  The cooperation among clouds  was further investigated in \cite{kaewpuang2013framework} to maximize the revenues of clouds and meet mobiles' demands via resource pool sharing.   Prior work on MECO resource allocation focuses on complex algorithmic designs and yields little insight into the optimal policy structures. In contrast, for a multiuser MECO system based on time-division multiple access (TDMA), the optimal resource-allocation policy is shown in the current work to have a simple threshold-based structure with respect to a derived offloading priority function. This insight is used for designing the low-complexity resource-allocation policy for a orthogonal frequency-division multiple access (OFDMA) MECO system.

Resource allocation for traditional multiple-access communication systems has been widely studied,  including TDMA (see e.g., \cite{wang2008power}),  OFDMA (see e.g., \cite{wong1999multiuser}) and code-division multiple access (CDMA) (see e.g., \cite{oh2003optimal}). Moreover, it has been designed for existing networks such as cognitive radio \cite{le2008resource} and heterogenous networks \cite{choi2010joint}. Note that all of them only focus on the radio resource allocation. In contrast, for the newly proposed MECO systems,  both the computation and radio resource allocation at the edge cloud are jointly  optimized for the maximum mobile energy savings, making the algorithmic design more complex.

\subsection{Contribution and Organization}
This paper considers resource allocation in  a multiuser MECO system based on TDMA and OFDMA. Multiple mobiles are required to compute different computation loads with the same latency constraint. Assuming that computation data can be split for separate computing, each mobile can simultaneously perform local computing and offloading. Moreover,  the edge cloud is assumed to have perfect knowledge of local computing energy consumption, channel gains and fairness factors at all users, which is used for designing centralized resource allocation to achieve the minimum weighted sum mobile energy consumption. In the TDMA MECO system, the optimal threshold-based policy is derived for both the cases of infinite and finite cloud capacities. For the OFDMA MECO system, a low-complexity sub-optimal  algorithm is proposed to solve the mixed-integer resource allocation problem.

The contributions of current work are  as follows.
\begin{itemize}
\item{\emph{TDMA MECO with infinite cloud capacity:} For TDMA MECO with infinite (computation) capacity, a convex optimization problem is formulated to minimize the weighted sum mobile energy consumption under the time-sharing constraint.  To solve it, an \emph{offloading priority function} is derived  that yields priorities for users and depends on their channel gains and local computing energy consumption. Based on this, the optimal  policy is proved to have a threshold-based structure that determines complete and minimum offloading for users with priorities above and below a given threshold, respectively.}
\item{\emph{TDMA MECO with finite cloud capacity:} The above results are extended to the case of finite capacity. Specifically, the optimal resource allocation policy is derived by defining an \emph{effective offloading priority function} and modifying the threshold-based policy as derived for the infinite-capacity cloud. To reduce the complexity arising from a two-dimension search for Lagrange multipliers, a simple and low-complexity algorithm is proposed based on the \emph{approximated} offloading priority order. This reduces the said search to a one-dimension search, shown by simulation to have close-to-optimal performance. }
\item{\emph{OFDMA MECO:} For a infinite-capacity cloud based on OFDMA, the insight of \emph{priority-based} policy structure of  TDMA is used for optimizing its resource allocation. Specifically, to solve the corresponding mixed-integer optimization problem, a low-complexity sub-optimal  algorithm is proposed. Using average sub-channel gains, the OFDMA  resource allocation problem  is transformed into its TDMA counterpart. Based on this, the initial resource allocation and offloaded data allocation can be determined by defining an \emph{average offloading priority function}. Moreover, the integer sub-channel assignment is performed according to the offloading priority order, followed by adjustments of offloaded data allocation over assigned sub-channels. The proposed algorithm is shown to have close-to-optimal performance by simulation and can be extended to the finite-capacity cloud case.  }
\end{itemize}

The reminder of this paper is organized as follows. Section II introduces the system model. Section III presents the problem formulation for multiuser MECO based on TDMA. The corresponding resource allocation policies are characterized in Section IV and Section V for both the cases of infinite and finite cloud capacities, respectively. The above results are extended in Section VI for the OFDMA system. Simulation results and discussion are given in Section VII, followed by the conclusion in Section VIII.

\section{System Model}\label{Sec:Sys}
Consider a multiuser MECO system shown in Fig.~\ref{Fig:Sys_Multiuser_MEC} with  $K$ single-antenna mobiles, denoted by a set $\mathcal{K}=\{1,2,\cdots, K\}$, and one single-antenna base station (BS) that is the gateway of an  edge cloud.  These mobiles are required to compute different computation loads under the same latency constraint. \footnote{For asynchronous computation offloading among users, the maximum additional latency for each user is one time slot. Moreover, this framework can be extended to predictive computing by designing control policies for the coming data. 
}
Assume that the BS has perfect knowledge of multiuser channel gains, local computing energy  per bit  and sizes of input data at all users, which can be obtained by feedback. Using these information, the BS selects offloading users, determines the offloaded data sizes and allocates radio resource to offloading users with the criterion of minimum weighted sum mobile energy consumption.

\subsection{Multiple-Access Model}\label{Section:MA}
 Both the TDMA and OFDMA systems are considered as follows. For the TDMA system, time is divided into slots each with a duration of $T$ seconds where $T$ is chosen to meet the user-latency requirement. As shown in Fig.~\ref{Fig:Sys_Multiuser_MEC}, each time slot comprises two sequential phases for 1) mobile offloading or local computing and 2) cloud computing and downloading of computation results from the edge cloud to mobiles. Cloud computing has small latency; the downloading  consumes negligible mobile energy and furthermore is much faster than offloading due to relative smaller sizes of computation results. For these reasons, the second phase is assumed to have a negligible duration compared to the first phase  and not considered in resource allocation. For the OFDMA system, the total bandwidth  is divided into multiple orthogonal sub-channels and each sub-channel can be assigned to  at most one user.  The offloading mobiles will be allocated with one or more sub-channels.

Considering an arbitrary  slot in TDMA/OFDMA, the BS schedules a subset of users for complete/partial offloading. The user with partial or no offloading computes a fraction of or all input data, respectively, using a local CPU.

\subsection{Local-Computing Model}
 Assume that the CPU frequency is fixed at each user and may vary over users. Consider an arbitrary time slot. Following the model in \cite{chen2015efficient}, let $C_k$ denote the number of CPU cycles required for computing $1$-bit of input data at the $k$-th mobile,  and $P_k$ the  energy consumption per  cycle for local computing at this user. Then the product $C_kP_k$ gives  computing energy per bit.  As shown in Fig.~\ref{Fig:Sys_MCC}, mobile  $k$ is required to compute $R_k$-bit input data within the time slot, out of which $\ell_k$-bit is offloaded and $(R_k -\ell_k)$-bit is computed locally.  Then the total energy consumption for local computing at mobile  $k$, denoted as $E_{\text{loc},k}$, is given by $E_{\text{loc},k}= (R_k-\ell_k)C_k P_k$. Let $F_k$  denote the computation capacity of  mobile $k$ that is measured by the number of CPU cycles per second. Under the computation latency constraint, it has $C_k (R_k-\ell_k)/F_k \le T.$ As a result, the offloaded data at mobile $k$ has the minimum size of   $\ell_k \ge m_k^+$ with $m_k=R_k-F_k T/C_k$, where $(x)^+=\max\{x, 0\}.$

\subsection{Computation-Offloading Model}
First, consider the TDMA system for an arbitrary time slot. Let $h_k$ denote the channel gain for mobile $k$ that is constant during offloading duration, and $p_k$ its transmission power.  Then the achievable rate (in bits/s), denoted by $r_k$, is:
\begin{equation}
r_k=B \log_2\l(1+\dfrac{p_k h_k^2}{N_0}\r) \label{Eq:TDMARate}
\end{equation} 
where  $B$ and $N_0$ are  the bandwidth and  the variance of complex white Gaussian channel noise, respectively. The fraction of slot allocated to mobile $k$ for offloading is denoted as $t_k$ with $t_k\geq 0$,   where $t_k =0$ corresponds to no offloading. For the case of offloading ($t_k > 0$), under the assumption of negligible cloud computing and result downloading time (see Section~\ref{Section:MA}), the transmission rate is fixed as $r_k=\ell_k/t_k$ since this is the most energy-efficient transmission policy under a deadline constraint \cite{PrabBiyi:EenergyEfficientTXLazyScheduling:2001}. Define a function $f(x)=N_0 (2^{\frac{x}{B}}-1)$. It follows from \eqref{Eq:TDMARate} that the energy consumption  for offloading at  mobile $k$ is
\begin{equation}
E_{\text{off},k}=p_k t_k=\dfrac{t_k}{h_k^2} f\!\l(\dfrac{\ell_k}{t_k}\r). \label{Eq:OffEgy}
\end{equation} 
Note that if either $\ell_k=0$ or $t_k = 0$, $E_{\text{off},k}$ is equal to zero.

Next, consider an OFDMA system with $N$ sub-channels, denoted by a set $\mathcal{N}=\{1,2,\cdots, N\}$.  Let $p_{k,n}$ and $h_{k,n}$ denote the transmission power and channel gain of mobile $k$ on the $n$-th sub-channel.   Define $\rho_{k,n} \!\!\in\! \{0,1\}$ as the sub-channel assignment indicator variable where $\rho_{k,n}=1$ indicates that   sub-channel $n$ is assigned to mobile $k$, and verse vice. Then the achievable rate (in bits/s) follows:
\begin{equation}\label{Eq:OFDMRate}
r_{k,n}= \rho_{k,n} \bar{B} \log_2\l(1+\dfrac{p_{k,n} h_{k,n}^2}{\bar{N_0}}\r) 
\end{equation} 
where $\bar{B}$ and $\bar{N_0}$ are the bandwidth and noise power for each sub-channel, respectively. Let  $\ell_{k,n}=r_{k,n} T$ denote the offloaded data size over the offloading duration time $T$ that can be set as the OFDMA symbol duration. The corresponding offloading  energy consumption can be expressed as below, which is similar to that in \cite{wong1999multiuser}, namely,
\begin{equation}
E_{\text{off},k,n}=\rho_{k,n}   p_{k,n} T=\dfrac{ \rho_{k,n}}{\bar{h}_{k,n}^2} \bar{f}\!\l(\dfrac{\ell_{k,n}}{ \rho_{k,n}}\r) \label{Eq:OffEgy2}
\end{equation} 
 where $\bar{h}_{k,n}^2=h_{k,n}^2/T$ and $\bar{f}(x)=\bar{N}_0 (2^{\frac{x}{\bar{B}T}}-1)$.

\begin{figure}[t]
\begin{center}
\includegraphics[width=8cm]{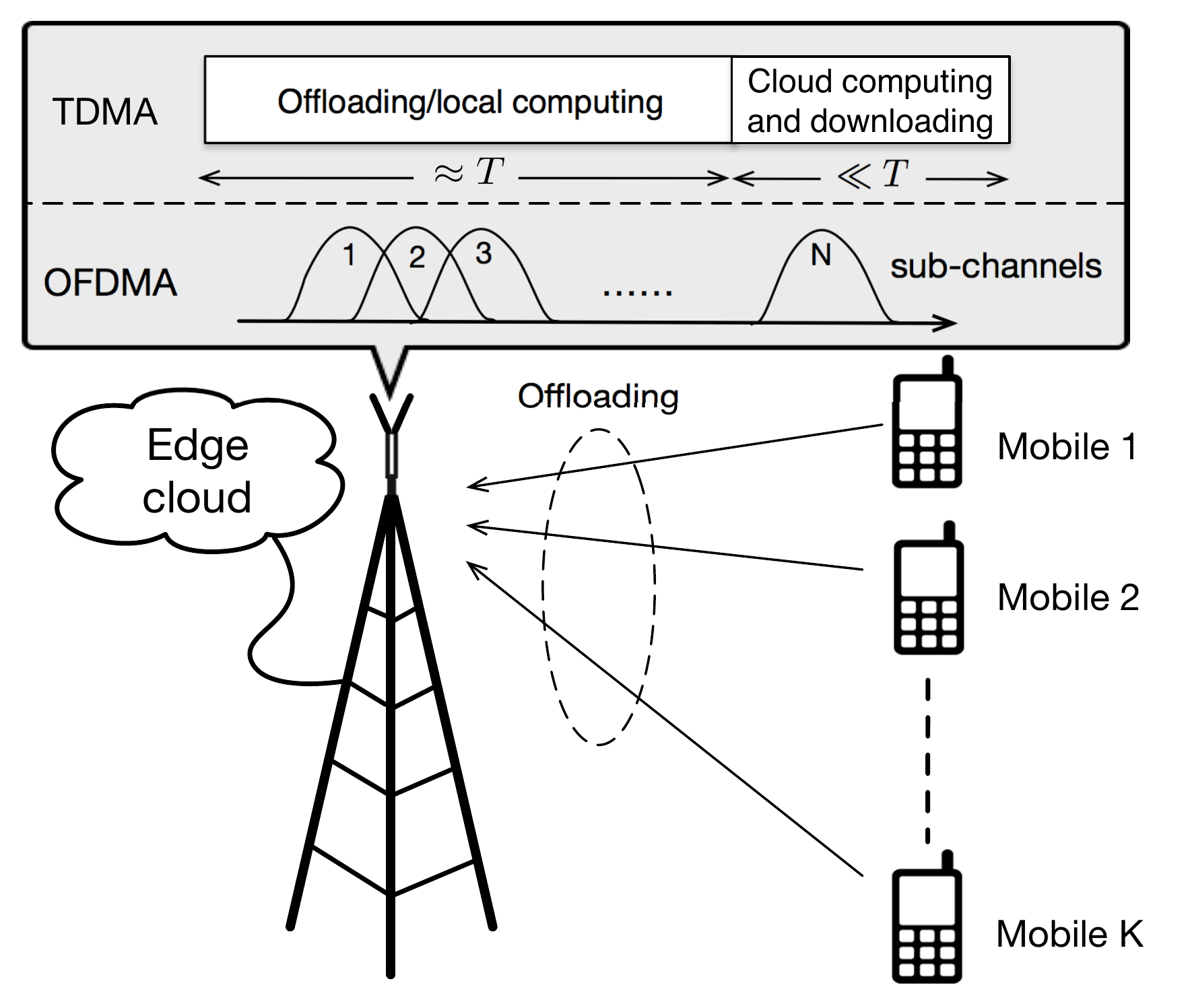}
\caption{Multiuser MECO systems based on TDMA and OFDMA.}
\label{Fig:Sys_Multiuser_MEC}
\end{center}
\end{figure}

\begin{figure}[t]
\begin{center}
\includegraphics[width=8cm]{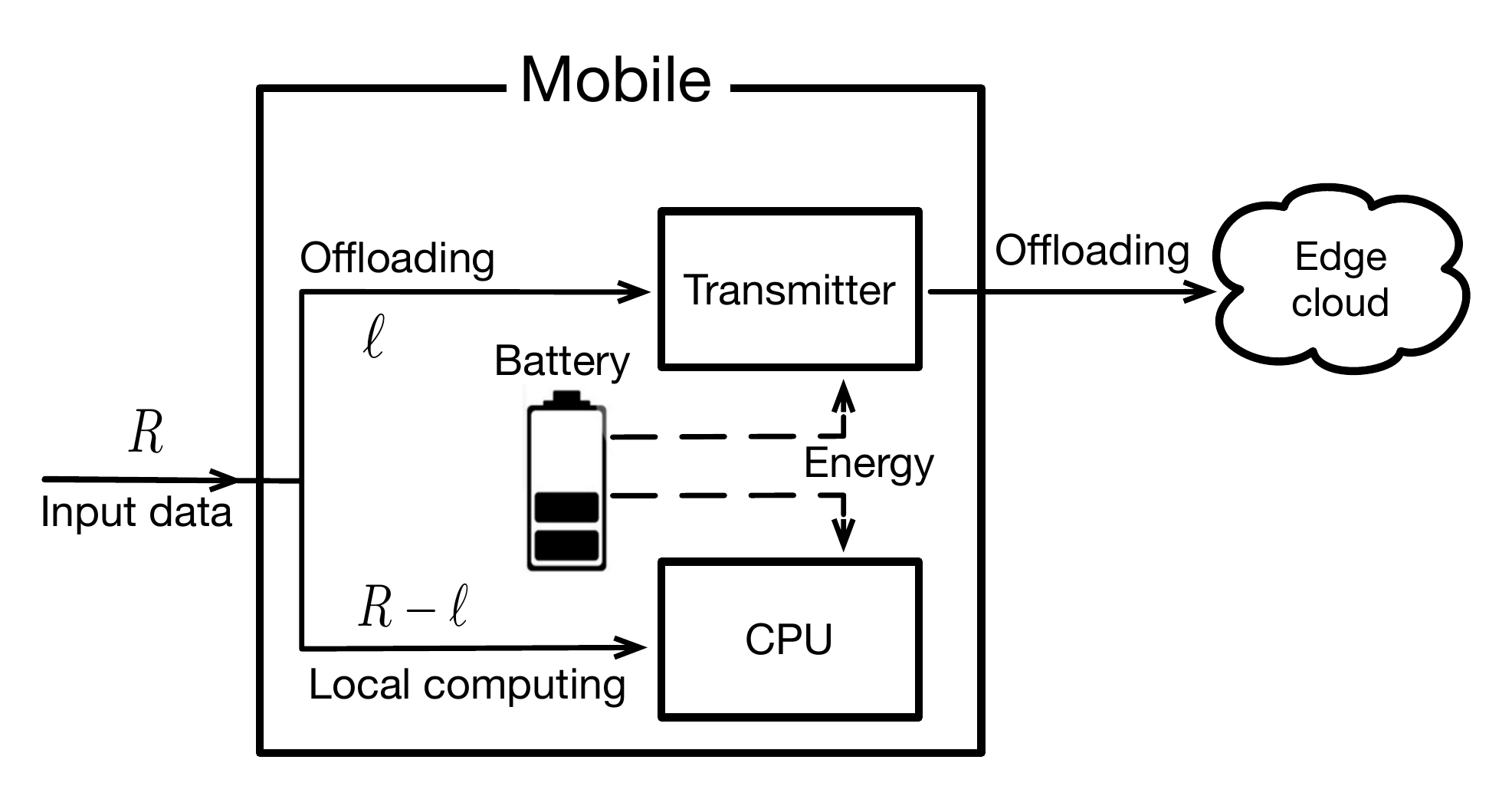}
\caption{Mobile computation offloading.}
\label{Fig:Sys_MCC}
\end{center}
\end{figure}

\subsection{Cloud-Computing Model}\label{Sec:CCM}
Considering an edge cloud with finite (computation) capacity, for simplicity, the finite capacity is reflected in one of the following two constraints. \footnote{For simplicity, we consider either a computation-load or a computation-time constraint at one time but not both simultaneously. However, note that the two constraints can be considered equivalent. Specifically, limiting the cloud computation load allows the computation to be completed within the required time and vice versa. The current resource-allocation policies can be extended to account for more elaborate constraints, which are outside the scope of the paper.} The first one upper-bounds CPU cycles of sum offloaded data that can be handled by the cloud in each time slot. Let $F$ represent the cloud computation capacity measured by CPU cycles per time slot. Then it follows: $\sum_{k=1}^{K} C_k \ell_k \le F $. This constraint ensures negligible  cloud computing latency. The other one considers non-negligible computing time at the cloud  that performs load balancing as in \cite{wang2010towards}, given as  $t_{\text{comp}}=(\sum_{k}^{K} \ell_k C_k)/F^{'},$ where $F^{'}$ is the cloud computation capacity measure by  CPU cycles per second. Note that $t_{\text{comp}}$ is factored into the latency constraint in the sequel.

\section{Multiuser  MECO for TDMA: \\Problem Formulation}
In this section, resource allocation for multiuser MECO based on TDMA is formulated as an optimization problem. The objective is to minimize the  weighted sum mobile energy consumption: $\sum_{k=1}^K \beta_k (E_{\text{off},k}+E_{\text{loc},k})$,  where the positive weight factors $\{\beta_k\}$ account for fairness among mobiles. Under the constraints on time-sharing, cloud computation capacity and computation latency, the resource allocation problem is formulated as follows: 
\begin{equation}\tag{$\textbf{P1}$} 
\begin{aligned}
\min_ {\{\ell_k, t_k\} }   ~ &\sum_{k=1}^{K} \beta_k \l[ \dfrac{t_k}{h_k^2}  f\!\l(\frac{\ell_k}{t_k}\r) + (R_k-\ell_k)C_k P_k \r] \\
\text{s.t.}\quad 
& \sum_{k=1}^K t_k \le T, \quad \sum_{k=1}^{K} C_k \ell_k \le F,\\ 
& t_k \ge 0, \quad m_k^+\le \ell_k \le R_k, \quad  \ k\in \mathcal{K}. 
\end{aligned}
\end{equation}

First, it is easy to observe that the feasibility condition for Problem P1 is: $\sum_{k=1}^K m_k^+ C_k \le F$. It shows that whether the cloud capacity constraint is satisfied  determines the feasibility of this optimization problem, while the time-sharing constraint can always be satisfied and only affects the mobile energy consumption. Next, one basic characteristic of Problem P1 is given in  the following  lemma, proved in Appendix~\ref{App:ConvexProblem}. 

\begin{lemma}\label{Lem:ConvexProblem} \emph{Problem P1 is a convex optimization problem.}
\end{lemma}

Assume that Problem P1 is feasible. The direct solution for Problem P1 using the dual-decomposition approach (the  Lagrange method)  requires iterative computation and yields little insight into the structure of the optimal policy. To address these issues, we adopt a two-stage solution approach  that requires first solving Problem P2  below, which  follows from  Problem P1 by relaxing the constraint on cloud  capacity: 
\begin{equation}\tag{$\textbf{P2}$} 
\begin{aligned}
\min_ {\{\ell_k, t_k\} }   ~ &\sum_{k=1}^{K} \beta_k \l[ \dfrac{t_k}{h_k^2}  f\!\l(\frac{\ell_k}{t_k}\r) + (R_k-\ell_k)C_k P_k \r] \\
\text{s.t.}\quad 
& \sum_{k=1}^K t_k \le T,\\ 
& t_k \ge 0,  \quad m_k^+\le \ell_k \le R_k, \quad  \ k\in \mathcal{K}. 
\end{aligned}
\end{equation}
If the solution for Problem P2 violates the constraint on cloud  capacity, Problem P1 is then incrementally solved building on the solution for Problem P2. This approach allows the optimal policy to be  shown to have the said  threshold-based structure and also facilitates the design of low-complexity close-to-optimal algorithm. It is interesting to note that Problem P2  corresponds to the case where the edge cloud has infinite capacity.  The detailed procedures for solving Problems P1 and P2 are presented in the two subsequent sections.

\section{Multiuser  MECO for TDMA:\\ Infinite Cloud Capacity}
In this section, by solving Problem P2 using the Lagrange method, we derive a threshold-based policy for the optimal resource allocation. 

To solve Problem P2, the  partial Lagrange function is defined as 
\begin{equation*}
L=\sum_{k=1}^{K} \beta_k \l[ \dfrac{t_k}{h_k^2}  f\l(\frac{\ell_k}{t_k}\r) + (R_k-\ell_k) C_k P_k \r]+ \lambda\l(\sum_{k=1}^K t_k - T\r)
\end{equation*}
where $\lambda \geq 0$ is the Lagrange multiplier associated with the time-sharing constraint. For ease of notation, define a function $g(x)=f(x)-x f^{'}(x)$. Let $\{\ell_k^{*(2)}, t_k^{*(2)}\}$ denote the optimal solution for Problem P2 that always exists satisfying 
the feasibility condition. 

Then applying KKT conditions leads to the following necessary and sufficient conditions:
\begin{subequations}
\begin{align}\label{Eq:OptL}
 &\dfrac{\partial L}{\partial \ell_k^{*(2)}}=\frac{\beta_k f^{'}\!\!\l(\frac{\ell_k^{*(2)}}{t_k^{*(2)}}\r)}{h_k^2}-\beta_k C_k P_k
\begin{cases}
>0, &\ell_k^{*(2)}=m_k^+\cr =0, \!&\ell_k^{*(2)} \in (m_k^+, R_k) \cr <0, &\ell_k^{*(2)}=R_k
\end{cases},\\
&\dfrac{\partial L}{\partial t_k^{*(2)}}=\dfrac{ \beta_k g\!\l(\frac{\ell_k^{*(2)}}{t_k^{*(2)}}\r)}{h_k^2}+\lambda^*
\begin{cases}
  >0, & t_k^{*(2)}=0 \cr =0, &t_k^{*(2)}>0 
\end{cases}, \quad \forall k \in \mathcal{K}, \label{Eq:OptT}\\
& \sum_{k=1}^K t_k^{*(2)}\le T, ~~~~~~~~~~~~~~\lambda^* \l(\sum_{k=1}^K t_k^{*(2)} - T\r)=0. \label{Eq:OptT:a}
\end{align}
\end{subequations}
Note that for $\ell^{*(2)}_k\in(m_k^+, R_k)$ and $t^{*(2)}_k>0$, it can be derived from \eqref{Eq:OptL} and \eqref{Eq:OptT} that   
\begin{equation}\label{Eq:LandT}
\dfrac{\ell_k^{*(2)}}{ t_k^{*(2)}}=f^{'-1}\!\l(C_k P_k h_k^2 \r)=g^{-1}\!\l(\frac{-h_k^2 \lambda^*}{\beta_k} \r).
\end{equation}

Based on these conditions,  the optimal policy for resource allocation is characterized in the following sub-sections. 
\subsection{Offloading Priority Function}
Define  a (mobile)  \emph{offloading priority function}, which is essential for the optimal resource allocation, as follows: 
\begin{equation}\label{Eq:OffPriority}
\phi(\beta_k, C_k,  P_k, h_k)=\begin{cases}
   \dfrac{\beta_k N_0}{h_k^2} \l( \upsilon_k  \ln \upsilon_k\!-\!\upsilon_k\!+\!1\r),  \!\!\!&\mbox{$\upsilon_k\ge 1$}\\
   0,  &\mbox{$\upsilon_k<1$}
   \end{cases},
\end{equation}
with the constant $\upsilon_k$ defined as 
\begin{equation}\label{Eq:Const}
\upsilon_k=\dfrac{B C_k  P_k h_k^2 }{N_0\ln2 }. 
\end{equation}
This function is derived by solving a useful equation as shown in the following lemma.
\begin{lemma}\label{Lem:PriStrongLA}\emph{
Given $\upsilon_k\ge1$, the offloading priority function $\phi(\beta_k, C_k,  P_k, h_k)$ in \eqref{Eq:OffPriority} is   the root of following equation with respect to $x$: $$f^{'-1}\!\l(C_k P_k h_k^2\r)=g^{-1}\!\l(\frac{-h_k^2 x}{\beta_k} \r).$$ 
}
\end{lemma}
Lemma~\ref{Lem:PriStrongLA} is proved in Appendix~\ref{App:PriStrongLA}. The function  generates an offloading priority value, $\phi_k = \phi(\beta_k, C_k,  P_k, h_k)$,   for mobile $k$,  depending on corresponding variables quantifying  fairness, local computing and channel. The amount of offloaded data by a mobile grows with an increasing offloading priority as shown in the next sub-section. It is useful to understand the effects of parameters on the offloading priority that are characterized as follows. 

\begin{lemma}\label{Lem:Phi}\emph{
Given $\upsilon \geq 1$, $\phi(\beta, C, P, h)$ is a \emph{monotone increasing function} for  $\beta$, $C$, $P$ and $h$.  
}
\end{lemma}

Lemma~\ref{Lem:Phi} is proved in Appendix~\ref{App:Phi}, by deriving the first derivatives of $\phi$ with respect to each parameter. This lemma is consistent  with the intuition that, to reduce  energy consumption  by offloading, the BS should schedule those  mobiles having  high computing energy  consumption per bit (i.e., large $C$ and $P$) or good channels (i.e., large $h$). 

\begin{remark}[Effects of Parameters on the Offloading Priority]\emph{It can be observed from \eqref{Eq:OffPriority} and \eqref{Eq:Const} that the offloading priority scales with local computing energy per bit $CP$ approximately as $(CP)\ln(CP)$ and with the channel gain $h$ approximately as $\ln h$. The former scaling is much faster than the latter. This shows that the  computing energy per bit is dominant  over the  channel on determining whether to offload. }
\end{remark}

\subsection{Optimal Resource-Allocation Policy}
Based on  conditions in \eqref{Eq:OptL}-\eqref{Eq:OptT:a} and Lemma~\ref{Lem:PriStrongLA}, the main result of this section is derived,  given in the following theorem.
\begin{theorem}[Optimal Resource-Allocation Policy]\label{Theo:OptiPolicyP2}\emph{Consider the case of infinite cloud computation capacity. The optimal policy solving Problem P2 has the following structure. 
\begin{enumerate}
\item If $\upsilon_k \leq 1 $ and the minimum offloaded data size $m_k^+=0$ for all $k$,   none of these  users performs offloading, i.e., $$\ell_k^{*(2)}=t_k^{*(2)}=0 \quad  k \in \mathcal{K}.$$ 
\item If there exists mobile $k$ such that $\upsilon_k > 1$ or $m_k^+>0$, for $k \in \mathcal{K}$, 
\begin{align*}
\ell_k^{*(2)}&
   \begin{cases}
   =m_k^+,   &\mbox{$\phi_k<\lambda^*$}\\
   \in [m_k^+, R_k],  &\mbox{$\phi_k=\lambda^*$}\\
   =R_k , &\mbox{$\phi_k>\lambda^*$}
   \end{cases},
   \end{align*}
and
 \begin{equation*}
t_k^{*(2)}=\frac{ \ln 2}{B\l[W_0\l(\frac{\lambda^{*} h_k^2/\beta_k-N_0}{N_0 e}\r)+1\r]} \times \ell_k^{*(2)}
\end{equation*}
where $W_0(x)$ is the Lambert function and $\lambda^*$ is the optimal value of the Lagrange multiplier satisfying the active time-sharing constraint: $\sum_{k=1}^{K} t_k^{*(2)} =  T$. 
\end{enumerate}
}
\end{theorem}
\begin{proof}
See Appendix~\ref{App:OptiPolicyP2}.
\end{proof}

Theorem~\ref{Theo:OptiPolicyP2} reveals that the optimal resource-allocation policy has a threshold-based structure when offloading saves energy. In other words,  since the exact case of $\phi_k=\lambda^*$ rarely occurs in practice, the optimal policy makes a \emph{binary offloading decision} for each mobile. Specifically, if the corresponding offloading priority exceeds a given threshold, namely $\lambda^*$, the mobile should offload all input data to the edge cloud; otherwise,  the mobile should offload only the minimum amount of data under the computation latency constraint. This result is consistent with the intuition that the  greedy method can lead to the optimal  resource allocation.  Note that there are two groups of users selected to  perform the minimum offloading. One is the group of users for which it has positive minimum offloading data, i.e., $m_k>0$,  and offloading cannot save energy consumption since they have bad channels or small local computing energy such that $\upsilon_k\le 1$ and $\phi_k=0$. The second group is the set of users for which offloading is energy-efficient, i.e., $\upsilon_k >1$, however, have  relatively small offloading priorities, i.e., $\phi_k < \lambda^*$; they cannot perform complete offloading due to the limited radio resource.

\begin{remark}[Offloading or Not?]\emph{For a conventional TDMA communication  system, continuous transmission by at least one  mobile  is always advantageous under the criterion of minimum sum energy consumption \cite{wang2008power}. However, this does not always hold for a TDMA  MECO system where no offloading for all users may be preferred  as shown in Theorem~\ref{Theo:OptiPolicyP2}. Offloading is not necessary expect for two cases.  First, there exists at least one   mobile whose input-data size is too large such that complete local computing fails to meet the latency constraint. Second,  some  mobile has a sufficient high value for the product  $ C_k P_k h_k^2$, indicating  that energy savings can be achieved by offloading because of high channel gain or large local computing energy consumption. }
\end{remark}

\begin{remark}[Offloading Rate]\emph{It can be observed from Theorem~\ref{Theo:OptiPolicyP2}  that the offloading rate, defined as $\ell_k^{*(2)}/t_k^{*(2)}$ for mobile $k$,  is determined only by the channel gain and fairness  factor while other factors, namely $C_k$ and  $P_k$, affect the offloading decision. The rate increases with a growing  channel gain and vice versa since a large channel gain supports a higher  transmission rate or reduces transmission power, making offloading desirable for reducing energy consumption. }
\end{remark}
\begin{algorithm}[t]
  \caption{Optimal Algorithm for Solving Problem P2.}
  \label{Alg:SLA:Opt}
  \begin{itemize}
\item{\textbf{Step 1} [Initialize]: 
Let $\lambda_\ell=0$ and $\lambda_h=\lambda_{\text{max}}$.
 According to Theorem~\ref{Theo:OptiPolicyP2}, obtain $T_{\ell}=\sum_{k=1}^{K} t_{k,\ell}^{*(2)}$  and  $T_{h}=\sum_{k=1}^{K} t_{k,h}^{*(2)}$, where $\{t_{k,\ell}^{*(2)} \}$ and $\{t_{k,h}^{*(2)}\}$ are the allocated fractions of slot for the cases of $\lambda_\ell$ and $\lambda_h$, respectively. }
\item{\textbf{Step 2} [Bisection search]:
While $T_{\ell} \neq T$ and  $T_{h} \neq T$, update $\{\lambda_\ell, \lambda_h\}$ as follows. \\(1) Define $\lambda_{m}=(\lambda_\ell+\lambda_h)/2$ and compute $T_{m}$. \\(2) If $T_m = T$, then $\lambda^*=\lambda_m$ and the optimal policy can be determined. Otherwise, if $T_{m}<T$, let $\lambda_{h}=\lambda_{m}$ and if $T_{m}>T$, let $\lambda_{\ell}=\lambda_{m}$.
}
\end{itemize}
  \end{algorithm}

\begin{remark}[Low-Complexity Algorithm]\emph{
The traditional method for solving Problem P2 is the block-coordinate descending algorithm which performs iterative optimization  of the two sets of variables,  $\{\ell_k\}$ and $\{t_k\}$,  resulting in high computation complexity.  In contrast,  by exploiting  the threshold-based structure of the optimal resource-allocation policy in Theorem~\ref{Theo:OptiPolicyP2}, the proposed solution approach, described in Algorithm~\ref{Alg:SLA:Opt}, needs to perform only a  \emph{one-dimension} search for  $\lambda^{*}$, reducing the computation complexity significantly. To facilitate the search, next lemma gives the range of $\lambda^{*}$, which can be easily proved from Theorem~\ref{Theo:OptiPolicyP2}.}
\end{remark}

\begin{lemma}\label{Lem:ProLambda}
\emph{When there is at least one offloading mobile,  $\lambda^*$ satisfies:
$0\le \lambda^* \le \lambda_{\max}=\max_k  \phi_k.$
}
\end{lemma}

Furthermore, with the assumption of infinite cloud capacity, the effects of finite radio resource (i.e., the TDMA time-slot duration) are characterized in the following two propositions  in terms of the number of offloading users, which can be easily derived using Theorem~\ref{Theo:OptiPolicyP2}.
\begin{proposition}[Exclusive Mobile Computation Offloading]\label{Pro:ExMCO}\emph{For TDMA MECO with offloading users, only one mobile can offload computation if  $T\le \dfrac{R_m}{B\log_2\l(\frac{B C_m P_m h_m^2}{N_0 \ln 2}\r)}$
 where $m=\arg \max_{k} \phi_k$.
}
\end{proposition}
It indicates  that short time slot limits the number of offloading users. From another perspective, it means that if the \emph{winner} user $m$ has excessive data, it will take up all the resource.  

\begin{proposition}[Inclusive Mobile Computation Offloading]\label{Pro:InMCO}\emph{
All offloading-desired mobiles (defined as for which, it has $\phi_k>0$) will completely offload computation  if 
\begin{align*}
T\ge  \sum_{k \in \mathcal{O}_1 }& \frac{R_k \ln 2}{B\l[W_0(\frac{\lambda_{\min} h_k^2/\beta_k-N_0}{N_0 e})+1\r] }\\ & + \sum_{k \in \mathcal{O}_2 }\frac{m_k^+ \ln 2}{B\l[W_0(\frac{\lambda_{\min} h_k^2/\beta_k-N_0}{N_0 e})+1\r] }
\end{align*}
where   $\mathcal{O}_1=\{k ~ | \phi_k>0\}$, $\mathcal{O}_2=\{k ~| \phi_k=0\}$ and $\lambda_{\min}=\min_{k\in \mathcal{O}_1} \phi_k$.}
\end{proposition}
Proposition~\ref{Pro:InMCO} reveals that when $T$ exceeds a given threshold, the offloading-desired mobiles for which offloading brings energy savings, will offload all computation to the cloud.
\begin{remark}[Which Resource is Bottleneck?]\emph{Proposition~\ref{Pro:ExMCO} and \ref{Pro:InMCO} suggest that as the radio resource continuously increases, the cloud will become the performance bottleneck and the assumption of infinite cloud capacity will not hold.
For a short time-slot duration, only a few users can offload computation. This just requires a fraction of computation   such that the cloud can be regarded as having infinite capacity. However, when  the time-slot duration is large, it not only saves energy consumption by offloading but also allows more  users for offloading, which potentially exceeds the cloud  capacity. The case of finite-capacity cloud will be considered in the sequel.}
\end{remark}

\subsection{Special Cases}
The optimal resource-allocation policies for several special cases considering equal fairness factors are discussed as follows. 
\subsubsection{Uniform  Channels and Local Computing} Consider the simplest case where $\{ h_k, C_k, P_k\}$ are identical for all $k$. Then all mobiles have uniform offloading priorities. In this case, for the optimal resource allocation, all mobiles can offload arbitrary data sizes so long as the sum offloaded data size satisfies the following constraint:  $\sum_{k=1}^K \ell_k^{*(2)} \leq T B\log_2\l(\dfrac{  B C P h^2}{N_0 \ln 2}\r).$

\subsubsection{Uniform Channels} Consider the case of   $h_1 = h_2\cdots = h_K=h$. The offloading priority for each mobile, say mobile $k$,   is only affected by the corresponding  local-computing parameters $P_k$ and $C_k$. Without loss of generality, assume that $P_1 C_1 \leq P_2 C_2 \cdots \leq P_K C_K$. Then the optimal resource-allocation policy is given in the following corollary of Theorem~\ref{Theo:OptiPolicyP2}. 

\begin{corollary}\label{Coro: SpeCase2}\emph{Assume infinite cloud capacity,  $h_1 = h_2\cdots = h_K=h$ and $P_1 C_1 \leq P_2 C_2 \cdots \leq P_K C_K$. Let $k_t$ denote the index such that $\phi_k < \lambda^\ast$ for all $k < k_t $ and $\phi_k >  \lambda^\ast$ for all $k \geq  k_t $, neglecting the rare case where $\phi_k = \lambda^\ast$. The optimal resource-allocation policy is given as follows:  for $k \in \mathcal{K}$,
\begin{align*}
\ell_{k}^{*(2)}&=
   \begin{cases}
   R_k,  & k \geq k_t\\
   m_k^+ , &\mbox{otherwise}
   \end{cases},
\end{align*}
and $$ t_k^{*(2)}=\frac{ \ln 2}{B\l[W_0\l(\frac{\lambda^{*} h^2/\beta-N_0}{N_0 e}\r)+1\r]}\times \ell_{k}^{*(2)}.$$
}
\end{corollary}
The result shows that the  optimal resource-allocation policy follows  a \emph{greedy} approach that selects  mobiles in a descending order of energy consumption per bit for complete offloading until the time-sharing duration is fully utilized.

\subsubsection{Uniform Local Computing} Consider the case of   $C_1P_1 = C_2P_2\cdots = C_K P_K$. Similar to the previous case, the optimal resource-allocation policy can be shown to follow the greedy approach that selects mobiles for complete offloading in the descending order of channel gains.

\section{Multiuser  MECO for TDMA:\\ Finite  Cloud   Capacity}
In this section, we consider the case of finite cloud  capacity and analyze the optimal resource-allocation policy for solving Problem P1. The policy is shown to also have a threshold-based structure as the infinite-capacity counterpart derived in the preceding section.  Both the optimal and sub-optimal algorithms are presented for policy computation. The results are extended to the finite-capacity cloud with non-negligible computing time.

\subsection{Optimal Resource-Allocation Policy}
To solve  the convex Problem P1, the  corresponding partial Lagrange function is written as 
\begin{align}\nn
\tilde{L}=\sum_{k=1}^{K} &\beta_k \l[ \dfrac{t_k}{h_k^2}  f\l(\frac{\ell_k}{t_k}\r) + (R_k-\ell_k) C_k P_k \r]\\&+ \lambda\l(\sum_{k=1}^K t_k - T\r) +  \mu \l(\sum_{k=1}^{K} C_k \ell_k - F\r)
\end{align}
where $\mu \geq 0 $  is the Lagrange multiplier associated with the cloud  capacity constraint. Using the above Lagrange function,   it is straightforward to show that the corresponding KKT conditions can be modified from their infinite-capacity counterparts in \eqref{Eq:OptL}-\eqref{Eq:OptT:a} by  replacing $P_k$ with $\tilde{P}_k=P_k-\mu$, called the \emph{effective computation energy per cycle}. The resultant \emph{effective offloading priority function}, denoted as $\tilde{\phi}_{k}$,  can be modified accordingly from that in \eqref{Eq:OffPriority} as 
\begin{equation}\label{Eq:EffectiveOffPriority}
\tilde{\phi}(\beta_k, C_k, P_k, h_k,\tilde{\mu}^*)=\begin{cases}
   \dfrac{\beta_k N_0}{h_k^2} \l( \tilde{\upsilon}_k  \ln \tilde{\upsilon}_k-\tilde{\upsilon}_k+1\r),  \!\!\!\!\!\!&\mbox{$\tilde{v}_k\ge 1$}\\
   0,  &\mbox{$\tilde{v}_k<1$}
   \end{cases},
\end{equation} 
where $\tilde{\upsilon}_k=\dfrac{B C_k  (P_k-\tilde{\mu}^*) h_k^2 }{N_0\ln2 }.$

Moreover, it can be easily derived that a cloud with smaller capacity $F$ leads to a larger Lagrange multiplier $\tilde{\mu}^*$. It indicates that compared with $\phi_k$ in \eqref{Eq:OffPriority} for the case of infinite-capacity cloud, the effective offloading priority function here is also determined by the cloud capacity. Based on above discussion, the main result of this section follows.

\begin{theorem}
\label{Theo:OptiPolicyP1}\emph{Consider the finite-capacity cloud with upper-bounded offloaded computation.  The optimal policy solving Problem P1 has the same structure as that in Theorem~\ref{Theo:OptiPolicyP2} and is expressed in terms of  the priority function $\tilde{\phi}_{k}$ in 
\eqref{Eq:EffectiveOffPriority} and optimized Lagrange multipliers $\{\tilde{\lambda}^*, \tilde{\mu}^*\}$. }
\end{theorem}

\begin{remark}[Variation of Offloading Priority Order]\emph{Since $\tilde{\mu}^* >0$, it has $\tilde{\phi}_k<\phi_k$ for all $k$. Therefore, the offloading priority order may be different with that of infinite-capacity cloud, due to the varying decreasing rates of offloading priorities. The reason is that the finite-capacity cloud should make the tradeoff between energy savings and computation burden. To this end, it will select mobiles for offloading that can save significant energy and require less computation
 for each bit of data.
 }
\end{remark}

Computing the threshold for the optimal resource-allocation policy requires  a \emph{two-dimension search} over the Lagrange multipliers  $\{\tilde{\lambda}^*, \tilde{\mu}^*\}$, described in Algorithm~\ref{Alg:WLA:Opt}. For an efficient search, it is useful to limit the range of $\tilde{\lambda}^*$ and $\tilde{\mu}^*$ shown as below, which can be easily proved.

\begin{lemma}\label{Lem:ProMu}
\emph{When there is at least one offloading mobile, the optimal Lagrange multipliers $\{\tilde{\lambda}^*, \tilde{\mu}^*\}$ satisfy:
\begin{align*}
 0\le \tilde{\lambda}^* \le \lambda_{\max}, ~\text{and}~  0\le \tilde{\mu}^* \le \mu_{\max}=\max_k \l\{P_k-\dfrac{N_0 \ln 2}{B C_k h_k^2}\r\}
\end{align*}
where $\lambda_{\max}$ is defined in Lemma~\ref{Lem:ProLambda}.
}
\end{lemma}

Note that $\tilde{\mu}^* = 0$ corresponds to the case of infinite-capacity cloud and $\tilde{\mu}^* = \mu_{\max}$ to the case where offloading yields no energy savings for any mobile.

\subsection{Sub-Optimal Resource-Allocation Policy}\label{Sec:SoptTDMA}
\begin{algorithm}[t]
  \caption{ Optimal Algorithm for Solving Problem P1.}
  \label{Alg:WLA:Opt}
  \begin{itemize}
\item {\textbf{Step 1} [Check solution for Problem P2]: Perform Algorithm~\ref{Alg:SLA:Opt}. If $\sum_{k=1}^K \ell_k^{*(2)} \le F$, the optimal policy is given in Theorem~\ref{Theo:OptiPolicyP2}. Otherwise, go to Step~$2$.
}
\item{\textbf{Step 2} [Initialize]: Let $\mu_\ell=0$ and $\mu_h=\mu_{\max}$.  Based on Theorem~\ref{Theo:OptiPolicyP1}, obtain $F_\ell=\sum_{k=1}^K C_k \ell_{k,\ell}^*$ and $F_h=\sum_{k=1}^K C_k \ell_{k,h}^*$, where $\{\ell_{k,\ell}^*\}$ and $\{\ell_{k,h}^*\}$ are the offloaded data sizes for $\mu_\ell$ and $\mu_h$, respectively, involving the one-dimension search for $\tilde{\lambda}^{*}$. }
\item{\textbf{Step 3} [Bisection search]: While $F_{\ell} \neq F$ and  $F_{h} \neq F$, update $\{\mu_\ell, \mu_h\}$ as follows.\\ (1) Define $\mu_{m}=(\mu_\ell+\mu_h)/2$ and compute $F_{m}$. \\(2) If $F_m = F$, then $\tilde{\mu}^*=\mu_m$  and the optimal policy can be determined. Otherwise, if $F_{m}<F$, let $\mu_{h}=\mu_{m}$ and if $F_{m}>F$, let $\mu_{\ell}=\mu_{m}$.}
\end{itemize}
\end{algorithm} 

\begin{algorithm}[t]
  \caption{ Sub-Optimal Algorithm for Solving  Problem P1.}
  \label{Alg:WLA:Sopt}
  \begin{itemize}
  \item{\textbf{Step 1}:  Perform Algorithm~\ref{Alg:SLA:Opt}. If $\sum_{k=1}^K \ell_k^{*(2)} \le F$, Theorem~\ref{Theo:OptiPolicyP2} gives the optimal policy. Otherwise, go to Step $2$.}
\item{\textbf{Step 2}: Based on  offloading priorities in \eqref{Eq:OffPriority}, offload the data from mobiles in the descending order of offloading  priority until  the cloud computation capacity is fully occupied, i.e., $\sum_{k=1}^{K} C_k \ell_k^* = F.$}
\item{\textbf{Step 3}: With $\{\ell_k^{*}\}$ derived in Step $2$, perform one-dimension search for $\lambda^{*}$ such that $\sum_{k=1}^K t_k^{*}=T$ where $t_k^{*}= \dfrac{\ell_k^{*}\ln 2}{B[W_0(\frac{\lambda^{*} h_k^2/\beta_k-N_0}{N_0 e})+1]}$.
}
\end{itemize}
  \end{algorithm}

To reduce the computation  complexity of Algorithm \ref{Alg:WLA:Opt} due to the two-dimension search,  one simple sub-optimal  policy is proposed as shown in Algorithm~\ref{Alg:WLA:Sopt}. The key idea is to decouple the computation and radio resource allocation.  In Step~$2$, based on the \emph{approximated} offloading priority in \eqref{Eq:OffPriority} for the case of infinite-capacity cloud, we allocate the computation resource to mobiles with high offloading priorities. Step $3$  optimizes  the  corresponding fractions of slot given offloaded data. {This sub-optimal algorithm has low computation complexity. Specifically, given a solution accuracy $\varepsilon>0$, the  iteration complexity for one-dimensional search can be given as $\mathcal{O}(\log(1/\varepsilon))$. For each iteration, the resource-allocation complexity is $\mathcal{O}(K)$. Thus, the total computation complexity for the sub-optimal algorithm is $\mathcal{O}(K\log(1/\varepsilon))$. Moreover,  its performance is shown by simulation  to be close-to-optimal in the sequel.

\subsection{Extension: MECO with Non-Negligible Computing Time }
Consider another finite-capacity cloud for which the computing time is non-negligible. Surprisingly, the resultant optimal policy is also threshold based,  with respect to a different offloading priority function.

Assume that the edge cloud performs load balancing for the uploaded computation as in \cite{wang2010towards}. In other words, the CPU cycles are proportionally allocated for each user such that all users experience the same computing time: $(\sum_{k=1}^K C_k \ell_k)/F^{'}$  (see Section~\ref{Sec:CCM}). Then the latency constraint is reformulated as $ (\sum_{k=1}^K C_k \ell_k)/F^{'}+\sum_{k=1}^K t_k \le T$, accounting for both the data transmission and  cloud computing time. The resultant optimization problem for minimizing  weighted sum mobile energy consumption is re-written by 
\begin{equation}\tag{$\textbf{P3}$}\label{Pro:CompuiingTime} 
\begin{aligned}
\min_ {\{\ell_k, t_k\} }   ~ &\sum_{k=1}^{K} \beta_k \l[ \dfrac{t_k}{h_k^2}  f\!\l(\frac{\ell_k}{t_k}\r) + (R_k-\ell_k)C_k P_k \r] \\
\text{s.t.}\quad 
& \frac{\sum_{k=1}^K C_k \ell_k}{F^{'}}+\sum_{k=1}^K t_k \le T,\\ 
& t_k \ge 0,  \quad m_k^+\le \ell_k \le R_k, ~~~ k \in \mathcal{K}.
\end{aligned}
\end{equation}

The key challenge of Problem P3 is that the amount of offloaded data size for each user has effects on offloading energy consumption, offloading duration and cloud computing time, making the problem more complicated. 

The feasibility condition for Problem P3 can be easily obtained as: $(\sum_{k=1}^K C_k m_k^+)/F^{'}<T.$ Note that the case $(\sum_{k=1}^K C_k m_k^+)/F^{'}=T$ makes Problem P3 infeasible since the resultant offloading time ($t_k=0$) cannot enable computation offloading.

Similarly, to solve Problem P3, the partial Lagrange function is written as 
\begin{align*}
\widehat{L}=\sum_{k=1}^{K} \beta_k &\l[ \dfrac{t_k}{h_k^2}  f\l(\frac{\ell_k}{t_k}\r) + (R_k-\ell_k) C_k P_k \r] \\&+ \lambda\l(\frac{\sum_{k=1}^K C_k\ell_k}{F^{'}}+\sum_{k=1}^K t_k - T\r).
\end{align*}
Define two sets of important constants: $a_k=\frac{F^{'} \ln 2}{B C_k}$ and $b_k=\frac{F^{'} P_k h_k^2}{N_0}$ for all $k$. Using KKT conditions, we can obtain the following offloading priority function
\begin{equation}\label{Eq:OffPriorityCompu}
\widehat{\phi}(\beta_k, C_k,  P_k, h_k, F^{'})\!=\!\!\begin{cases}
   \dfrac{\beta_k N_0}{h_k^2} \l( \widehat{\upsilon}_k  \ln \widehat{\upsilon}_k-\widehat{\upsilon}_k\!+\!1\r),  &\mbox{$\widehat{\upsilon}_k\ge 1$}\\
   0,  &\mbox{$\widehat{\upsilon}_k<1$}
   \end{cases},
\end{equation}
where 
\begin{equation}\label{Eq:Upsilon3}
\widehat{\upsilon}_k=\dfrac{b_k-1}{ W_0((b_k-1)e^{(a_k-1)})}.
\end{equation}

This function is derived by solving a equation in the following lemma, proved in Appendix~\ref{App:OffPrioNLA}.

\begin{lemma}\label{Lem:OffPrioNLA}\emph{Given $\widehat{\upsilon}_k\ge1$, the offloading priority function $\widehat{\phi}_k=\widehat{\phi}(\beta_k, C_k,  P_k, h_k,F^{'})$ in \eqref{Eq:OffPriorityCompu} is  the root of the following equation with respect to $x$: 
\begin{equation}\label{Eq:FunPriCloud2}
f^{'-1}\!\l(C_k P_k h_k^2-\frac{x C_k h_k^2}{\beta_k F^{'} }\r)=g^{-1}\!\l(\frac{-h_k^2 x}{\beta_k} \r).
\end{equation}
}
\end{lemma}

Recall that for a  cloud that upper-bounds the offloaded computation, its offloading priority (i.e., $\tilde{\phi}_k$ in 
\eqref{Eq:EffectiveOffPriority}) is function of a Lagrange multiplier $\tilde{\mu}^*$ which is determined by $F$. However, for the current cloud with non-negligible computing time, the offloading priority function $\widehat{\phi}_k$ in \eqref{Eq:OffPriorityCompu} is directly affected by the finite cloud capacity $F^{'}$ via $\widehat{\upsilon}_k$.

 In the following, the properties of $\widehat{\upsilon}_k$, which is the key component of $\widehat{\phi}_k$, are characterized. 

\begin{lemma}\label{Lem:WandAB}\emph{
$\widehat{\upsilon}>1$ if and only if $ \upsilon>1$,
 where $\upsilon$ is defined in \eqref{Eq:Const}.
}
\end{lemma}
It is proved in Appendix~\ref{App:WandAB} and indicates that the condition that offloading saves energy comsumption for this kind of finite-capacity cloud is same as that of infinite-capacity cloud.

\begin{lemma}\label{Lem:PhiCloudII}\emph{
Given $\widehat{\upsilon} \ge 1$, $\widehat{\phi}(\beta, C, P, h, F^{'})$ is a monotone increasing  function for $\beta$, $C$, $P$, $h$ and $F^{'}$, respectively.}
\end{lemma}

Similar to Lemma~\ref{Lem:Phi}, Lemma~\ref{Lem:PhiCloudII} can be proved by deriving the first derivatives of $\widehat{\phi}$ with respect to each parameter.  It shows that enhancing the cloud capacity will increase the offloading priority for all users that is same as the result of a cloud with upper-bounded offloaded computation.

Based on above discussion, the main result of this section are presented in the following theorem.  
\begin{theorem}\label{Theo:OptiPolicyP3}\emph{Consider the finite-capacity cloud with non-negligible computing time. The optimal resource allocation policy solving Problem P3 has the same structure as that in Theorem~\ref{Theo:OptiPolicyP2} and is expressed in terms of  the priority function $\widehat{\phi}_{k}$ in 
\eqref{Eq:OffPriorityCompu} and optimized Lagrange multipliers $\widehat{\lambda}^*$. }
\end{theorem} 

The optimal policy can be computed with a one-dimension search for $\widehat{\lambda}^*$, following a similar procedure in Algorithm~\ref{Alg:SLA:Opt}.

\section{Multiuser MECO for OFDMA}
In this section, consider  resource allocation for MECO OFDMA. Both OFDM sub-channels and offloaded data sizes are optimized for the energy-efficient multiuser MECO. To solve the formulated  mixed-integer  optimization problem, a sub-optimal algorithm is proposed by defining an average offloading priority function from its TDMA counterpart and shown to have close-to-optimal performance in simulation.

\subsection{Multiuser MECO for OFDMA: Infinite Cloud Capacity}\label{Sec:OFDMAInfinite}
Consider an OFDMA system (see Section~\ref{Sec:Sys}) with $K$ mobiles and $N$ sub-channels. The cloud is assumed with infinite cloud capacity. Given time-slot duration $T$, the latency constraint for local computing  is rewritten as $C_k (R_k-\sum_{n=1}^N \ell_{k,n})/F_k \le T$. Moreover, the time-sharing constraint is replaced by sub-channel constraints, expressed as $\sum_{k=1}^K \rho_{k, n} \le 1$ for all $n$. Then the corresponding optimization problem 
for the minimum weighted sum mobile energy consumption based on OFDMA is readily re-formulated as: 
\begin{equation}\tag{$\textbf{P4}$} 
\begin{aligned}
\min_ {\{\ell_{k,n}, \rho_{k,n}\} }  &\sum_{k=1}^{K} \beta_k \l[\sum_{n=1}^N  \dfrac{\rho_{k,n}}{\bar{h}_{k,n}^2}  \bar{f}\!\l(\frac{\ell_{k,n}}{\rho_{k,n}}\r) + (R_k\!-\!\sum_{n=1}^N \ell_{k,n})C_k P_k \r] \\
\text{s.t.}\qquad 
&  \sum_{k=1}^K \rho_{k,n} \le 1,  ~~~~~~~~~~~~~n\in\mathcal{N}; \\
 & m_k^+\le \sum_{n=1}^ N \ell_{k,n} \le R_k,  ~~~~k \in \mathcal{K};\\
& \rho_{k,n}\in \{0,1\},  ~~~~~~~~~~~~~n\in \mathcal{N} ~\text{and}~ k\in \mathcal{K}.
\end{aligned}
\end{equation}

Observe that Problem P4 is a  mixed-integer programming problem that  is difficult to solve. It involves the joint optimization of both continuous variables $\{\ell_{k,n}\}$ and integer variables $\{\rho_{k,n}\}$. One common solution method is relaxation-and-rounding, which firstly relaxes the integer constraint $\rho_{k,n}\in \{0,1\}$ as the real-value constraint $0\le \rho_{k,n} \le 1$ \cite{wong1999multiuser}, and then determines the integer solution using rounding techniques. Note that the integer-relaxation problem is a convex problem which can be solved by powerful convex optimization techniques. An alternative method is using dual decomposition as in \cite{tao2008resource}, which has been proved to be optimal when the number of sub-channels goes to infinity. However, both algorithms performing extensive iterations shed little insight on the policy structure. 

To reduce the  computation complexity and characterize the  policy structure, a low-complexity sub-optimal algorithm is proposed below by a decomposition method, motivated by the following existing results and observations. First, for traditional OFDMA systems, low-complexity sub-channel allocation policy was designed in \cite{huang2009joint,kivanc2003computationally} via defining average channel gains, which was shown to  achieve close-to-optimal performance in simulation. Next, for the integer-relaxation resource allocation problem, applying KKT conditions directly can lead to its optimal solution. It can be observed that for each sub-channel, users with higher offloading priorities should be allocated with more radio resource. Therefore, in the proposed algorithm, the initial resource and offloaded data allocation is firstly determined by defining average channels gains and an \emph{average offloading priority function}. Then, the integer sub-channel assignment is performed according to the offloading priority order, followed by the adjustment of offloaded data allocation over assigned sub-channels for each user.  The main procedures of this  \emph{sequential} algorithm are  as follows. 
\begin{enumerate}
\item[--]{\textbf{Phase 1}  [Sub-Channel Reservation for Offloading-Required Users]: Consider the offloading-required users that have $m_k^+>0$. The offloading priorities for these users are ordered in the descending manner.  Based on this, the available sub-channels with high priorities are assigned to corresponding users sequentially and each user is allocated with  one sub-channel.}
\item[--]{\textbf{Phase 2}  [Initial Resource and Offloaded Data Allocation]:  For the unassigned sub-channels, using average channel gain over these sub-channels for each user, the OFDMA MECO problem is transformed into its TDMA counterpart. Then, by defining an average offloading priority function, the optimal total sub-channel number and offloaded data size for each user are derived. Note that the resultant sub-channel numbers may not be integer. }
\item[--]{\textbf{Phase 3} [Integer Sub-Channel Assignment]: Given constraints on the rounded total sub-channel numbers for each user derived in Phase $2$, specific integer sub-channel assignment is determined by the offloading priority order. Specifically, each sub-channel is assigned to the user that requires sub-channel assignment and has higher offloading priority than others.
}  
\item[--]{\textbf{Phase 4} [Adjustment of Offloaded Data Allocation]: For each user, based on the sub-channel assignment in Phase $3$, the specific offloaded data allocation is optimized. }
\end{enumerate}

Before stating the algorithm, let  $\phi_{k,n}$ define the offloading priority function for user $k$ at sub-channel $n$. It can be modified from the TDMA counterpart in \eqref{Eq:OffPriority} by replacing  $h_k$, $N_0$ and $\upsilon_k$ with $h_{k,n}$, $\bar{N}_0$ and $\upsilon_{k,n}=\frac{\bar{B} T C_k  P_k \bar{h}_{k,n}^2 }{ \bar{N}_0\ln 2 }$, respectively.  Let $\Phi$ reflect the offloading priority order, which is constituted by $\{\phi_{k,n}\}$, arranged in the descending manner, e.g., $\{ \phi_{2,3}\ge \phi_{1,4} \ge \cdots \phi_{5,2}\}$. The set of offloading-required users is denoted by $\mathcal{K}_1$, given as $\mathcal{K}_1=\{k, ~| m_k^+>0\}$. The sets of assigned and unassigned sub-channels  are denoted by $\mathcal{N}_1$ and $\mathcal{N}_2$, initialized as $\mathcal{N}_1=\emptyset$ and $\mathcal{N}_2=\mathcal{N}$. For each user, say user $k$, the assigned sub-channel set is represented by $\mathcal{S}_k$, initialized as  $\mathcal{S}_k=\emptyset$. In addition, sub-channel assignment indicators are set as $\{\rho_{k,n}=0\}$ at the beginning. 

Using these definitions, the detailed control policies are elaborated as follows. 

\subsubsection{Sub-Channel Reservation for Offloading-Required Users}
 The purpose of this phase is to guarantee that the computation latency constraints are satisfied for all users. This can be achieved  by reserving one sub-channel for each offloading-required user as presented in Algorithm~\ref{Alg:SubReser}.

Observe that Step~$1$ in the loop searches for the highest offloading priority $\phi_{k',n'}$ over unassigned sub-channels $\mathcal{N}_2$ for the remaining offloading-required users $\mathcal{K}_1$; and then allocates  sub-channel $n'$ to user $k'$. This sequential sub-channel assignment follows the descending offloading priority order. Moreover, the condition for the loop ensures that all offloading-required users will be allocated with one sub-channel. This phase only has a complexity of $\mathcal{O}(K)$ since it just performs the $\max$ operation for at most $K$ iterations.
\begin{algorithm}[t]
  \caption{Sub-Channel Reservation for Offloading-Required Users.}
  \label{Alg:SubReser}
  While  $\mathcal{K}_1 \neq \emptyset$, reserve sub-channels as follows.
  \begin{itemize}
  \item[(1)] Let $\rho_{k',n'}=1$ where $\{k',n'\}=\underset{k \in \mathcal{K}_1,n\in\mathcal{N}_2}{\arg\max} \phi_{k,n}$. ~~~
  \item[(2)] Update sets: ~$\mathcal{S}_{k'}=\mathcal{S}_{k'} \cup \{n'\}$; ~$\mathcal{K}_1=\mathcal{K}_1\setminus \{ k'\}$; ~~~
  $\mathcal{N}_1= \mathcal{N}_1 \cup \{n'\}$; ~~~$\mathcal{N}_2=\mathcal{N} \setminus \mathcal{N}_1$.
  \end{itemize}
  \end{algorithm}

\subsubsection{Initial Resource and Offloaded Data Allocation} 
This phase determines the \emph{total} allocated sub-channel number and offloaded data size for each user. Note that the integer constraint on sub-channel allocation makes Problem P4  challenging, which requires an exhaustive search. To reduce the computation complexity, we first derive the non-integer total number of sub-channels for each user as below.

Using a similar method in \cite{kivanc2003computationally}, for each user, say user $k$, let $H_k$ denote its average sub-channel gain, give by $H_k = \sqrt{(\sum_{n\in \mathcal{N}_2} \bar{h}_{k,n}^2)/|\mathcal{N}_2|}$ where $|\mathcal{N}_2|$ gives the cardinality of unassigned sub-channel set $\mathcal{N}_2$ resulted from Phase $1$. Then, the MECO OFDMA resource allocation Problem P4 is transformed into its TDMA counterpart Problem P5 as:
\begin{equation}\tag{$\textbf{P5}$} 
\begin{aligned}
\min_ {\{\ell_k, n_k\} }   ~ &\sum_{k=1}^{K} \beta_k \l[ \dfrac{n_k}{H_k^2}  \bar{f}\!\l(\frac{\ell_k}{n_k}\r) + (R_k-\ell_k)C_k P_k \r] \\
\text{s.t.}\quad 
& \sum_{k=1}^K n_k \le |\mathcal{N}_2|,\\ 
& n_k \ge 0,  \quad m_k^+\le \ell_k \le R_k, ~~~  k \in \mathcal{K} 
\end{aligned}
\end{equation}
where $\{\ell_k, n_k\}$ are the allocated total sub-channel numbers and offloaded data sizes.

 Define an \emph{average offloading priority function} as in \eqref{Eq:OffPriority} by replacing $h_k$ with $H_k$.  The optimal control policy, denoted by $\{\ell_k^*, n_k^*\}$, can be directly obtained following the same method as for Theorem~\ref{Theo:OptiPolicyP2}.  Note that  this phase only invokes the bisection search.  Similar to Section~\ref{Sec:SoptTDMA}, the computation complexity can be represented by $\mathcal{O}(K\log{(1/\varepsilon)})$.

\subsubsection{Integer Sub-Channel Assignment} Given the non-integer total sub-channel number allocation obtained in Phase $2$, in this phase, users are assigned with specific integer sub-channels based on offloading priority order.  Specifically, it includes the following  two steps as in Algorithm~\ref{Alg:SSA}.

 In the first step, to guarantee that sub-channels are enough for allocation, each user is allocated with  $\tilde{n}_k^*=\lfloor n_k^{*}\rfloor$ sub-channels.  However, allocating specific sub-channels to users given the rounded numbers is still hard, for which the optimal solution can be obtained  using the Hungarian Algorithm \cite{kuhn1955hungarian} that has the complexity of $\mathcal{O}(N^3)$. To further reduce the complexity, a priority-based sub-channel assignment is proposed as follows. Let $\tilde{\mathcal{K}}$ denote the set of users that require sub-channel assignment, which is initialized as $\tilde{\mathcal{K}}=\{k,~| \tilde{n}_k^{*}>0 \}$ and will be updated as in Step $1.(3)$, by deleting the user that has been allocated with the maximum sub-channels.  During the loop, for users in set $\tilde{\mathcal{K}}$ and available sub-channels $\mathcal{N}_2$, we search for the highest offloading priority function, indexed as $\phi_{k',n'}$, and assign sub-channel $n'$ to user $k'$.
 
 In the second step, all users compete for  remaining sub-channels since  $\tilde{n}_k^*$ is the lower-rounding of  $n_k^{*}$ in the first step.  In particular, each unassigned sub-channel in $\mathcal{N}_2$ is assigned to the user with highest offloading priority.  In total, the computation complexity of this phase is $\mathcal{O}(N)$.

  \begin{algorithm}[t!]
  \caption{Integer Sub-Channel Assignment.}
  \label{Alg:SSA}
  \textbf{Step 1}: While  $\tilde{\mathcal{K}} \neq \emptyset$, assign sub-channels as follows.
    \begin{itemize}
 \item[(1)]  Let $\rho_{k',n'}=1$ where $\{k',n'\}=\underset{k \in \tilde{\mathcal{K}}, n\in\mathcal{N}_2}{\arg\max} \phi_{k,n}$.
 \item[(2)]  Update sets: ~~~$\mathcal{S}_{k'}=\mathcal{S}_{k'} \cup \{n'\}$;~~~~ $\mathcal{N}_1= \mathcal{N}_1 \cup \{n'\}$; ~~~~$\mathcal{N}_2=\mathcal{N} \setminus \mathcal{N}_1$. 
 \item[(3)] If $|\mathcal{S}_{k'}| = \tilde{n}_{k^{'}}^*$, then $\tilde{\mathcal{K}}=\tilde{\mathcal{K}} \setminus \{k'\}$.
 \end{itemize}
 \textbf{Step 2}: If $\mathcal{N}_2 \neq \emptyset$, assign remaining sub-channels as follows.
 ~For each $n\in \mathcal{N}_2$,  let $\rho_{k',n}=1$ where $k'= \arg\max_{k \in \mathcal K} 
  \phi_{k,n}$.
  \end{algorithm}
  
\subsubsection{Adjustment of Offloaded Data Allocation}
Based on results from Phase  $1$--$3$, for each user, say $k$, this phase allocates the total offloaded data $\ell_k^*$ over assigned sub-channels $\mathcal{S}_k$ for minimizing the individual mobile energy consumption. The corresponding optimization problem is formulated as below with the solution given in Proposition~\ref{Pro:SubDataAll}.
\begin{equation}\tag{$\textbf{P6}$} 
\begin{aligned}
\min_ {\{\ell_{k,n}\} }   ~ &\sum_{n\in \mathcal{S}_k}  \dfrac{1}{\bar{h}_{k,n}^2}  \bar{f}\!\l(\ell_{k,n}\r) \\
\text{s.t.}\quad 
 &\sum_{n \in \mathcal{S}_k } \ell_{k,n}= \ell_k^*, \\
& \ell_{k,n} \ge0, ~\ n\in \mathcal{S}_k .
\end{aligned}
\end{equation}

\begin{proposition}\label{Pro:SubDataAll}\emph{For user $k$, the optimal offloaded data allocation solving Problem P6 is
$$\ell_{k,n}^*=\l[\bar{B} T \log_2 \l(\dfrac{ \xi_k \bar{B}T  \bar{h}_{k,n}^2}{\bar{N}_0 \ln 2}\r)\r]^+ \quad \text{for}~ n \in \mathcal{S}_k$$ where $\xi_k$ satisfies $\sum_{n \in \mathcal{S}_k } \ell_{k,n}^*=\ell_k^*$.}
\end{proposition}
Note that it is possible that some sub-channels are allocated to user $k$ but without offloaded data allocation due to their poor sub-channel gains. For each user, the optimal solution is obtained by performing one-dimension search for $\xi_k$, whose computation complexity is $\mathcal{O}(N \log{(1/\varepsilon)})$ since $|\mathcal{S}_k|\le N$. Thus, the total complexity of this phase is  $\mathcal{O}(KN \log{(1/\varepsilon)})$, considering offloaded data allocation for all users. 

\begin{remark}[Low-Complexity Algorithm]\label{Rem:OFDMComplexity}\emph{Based on above discussion, the total complexity for the proposed sequential sub-optimal algorithm is up to $\mathcal{O}(K+N+KN \log{(1/\varepsilon)})$. It significantly reduces the computation complexity compared with that of relaxation-and-rounding  policy, which has complexity up to $\mathcal{O}((KN)^{3.5} \log{(1/\varepsilon)}+N)$  solved by CVX \cite{ben2001lectures}.}
\end{remark}

\begin{figure}[t]
\begin{center}
\includegraphics[width=7cm]{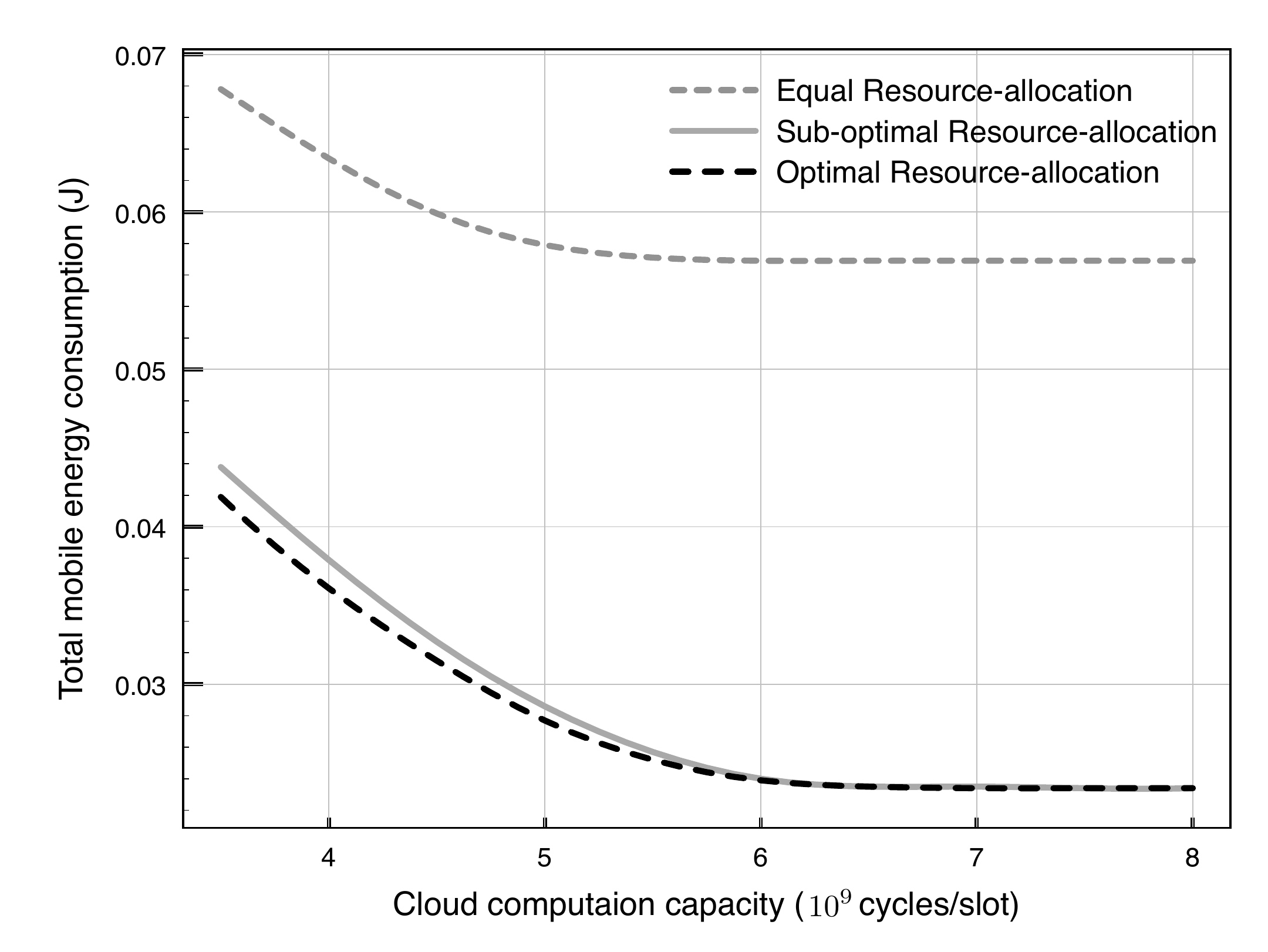}
\caption{Total mobile energy consumption vs. cloud computation capacity  for a TDMA system.}
\label{Fig:egy_vs_capacity}
\end{center}
\end{figure}

\subsection{Multiuser MECO for OFDMA: Finite Cloud Capacity}\label{Sec:OFDMAFinite}
For the case of finite-capacity  cloud based on OFDMA, the corresponding sub-optimal low-complexity algorithm can be derived by modifying that of infinite-capacity cloud  as follows. 

Recall that for TDMA MECO, modifying the offloading priority function of infinite-capacity cloud leads to the optimal resource allocation for the finite-capacity cloud. Therefore, by the similar method, modifying Phase $2$ to account for the finite computation capacity will give the new optimal initial resource and offloaded data allocation for all users. Other phases in Section~\ref{Sec:OFDMAInfinite} can be straightforwardly extended to the current case and are omitted for simplicity.
\section{Simulation Results}

\begin{figure*}[t!]
\centering
\subfigure[Effect of the time slot duration.]{\label{Fig:egy_vs_T}
\includegraphics[width=7cm]{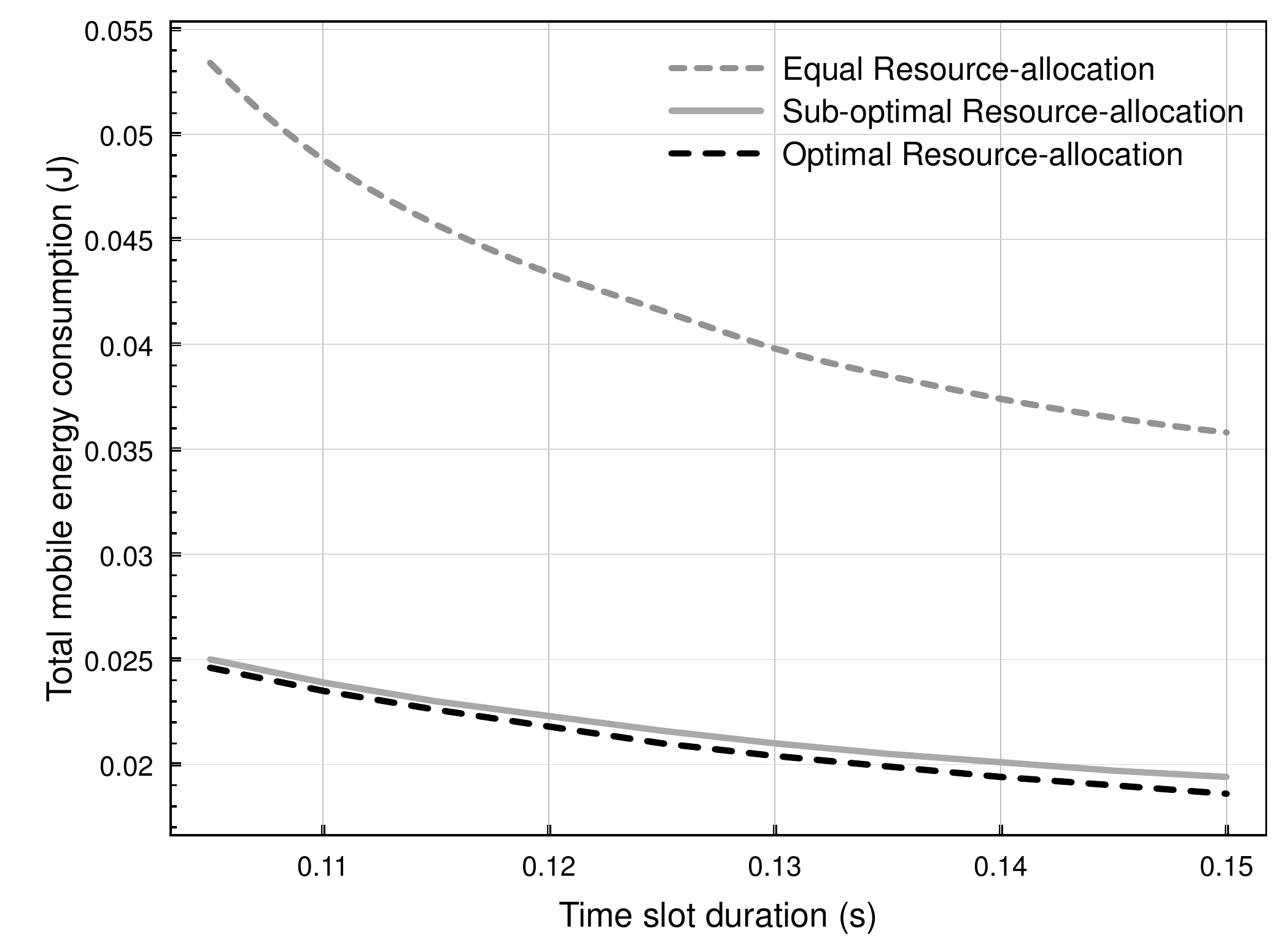}}
\hspace{15mm}
\subfigure[Effect of the number of users.]{\label{Fig:egy_vs_mobiles}
\includegraphics[width=7cm]{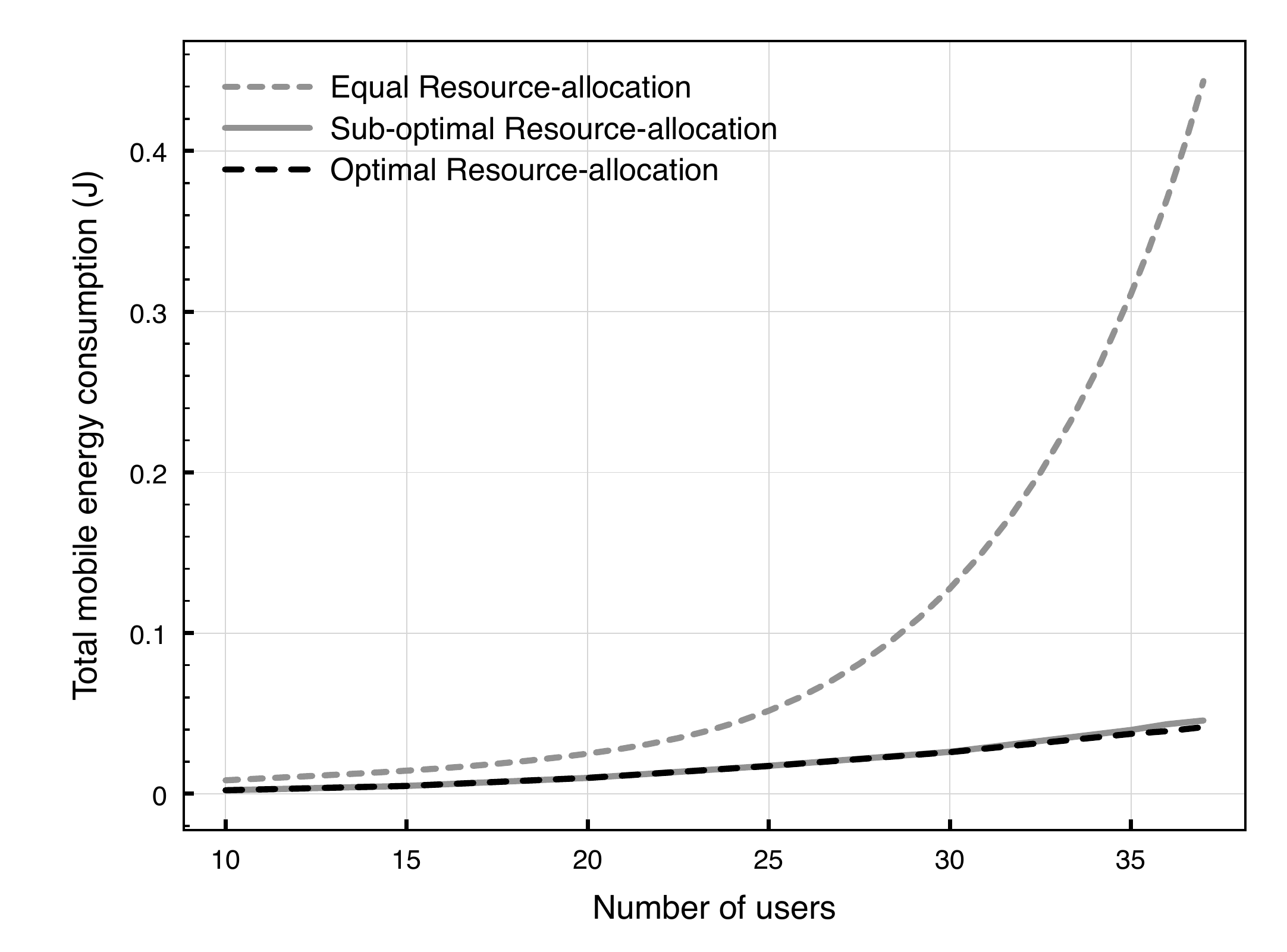}}
\caption{(a) Total mobile energy consumption vs. time slot duration for a TDMA system. (b) Total mobile energy consumption vs. number of users for a  TDMA system.}
\end{figure*}

\begin{figure*}[t!]
\centering
\subfigure[Effect of the number of sub-channels. ]{\label{Fig:egy_vs_sub_ofdm}
\includegraphics[width=7cm]{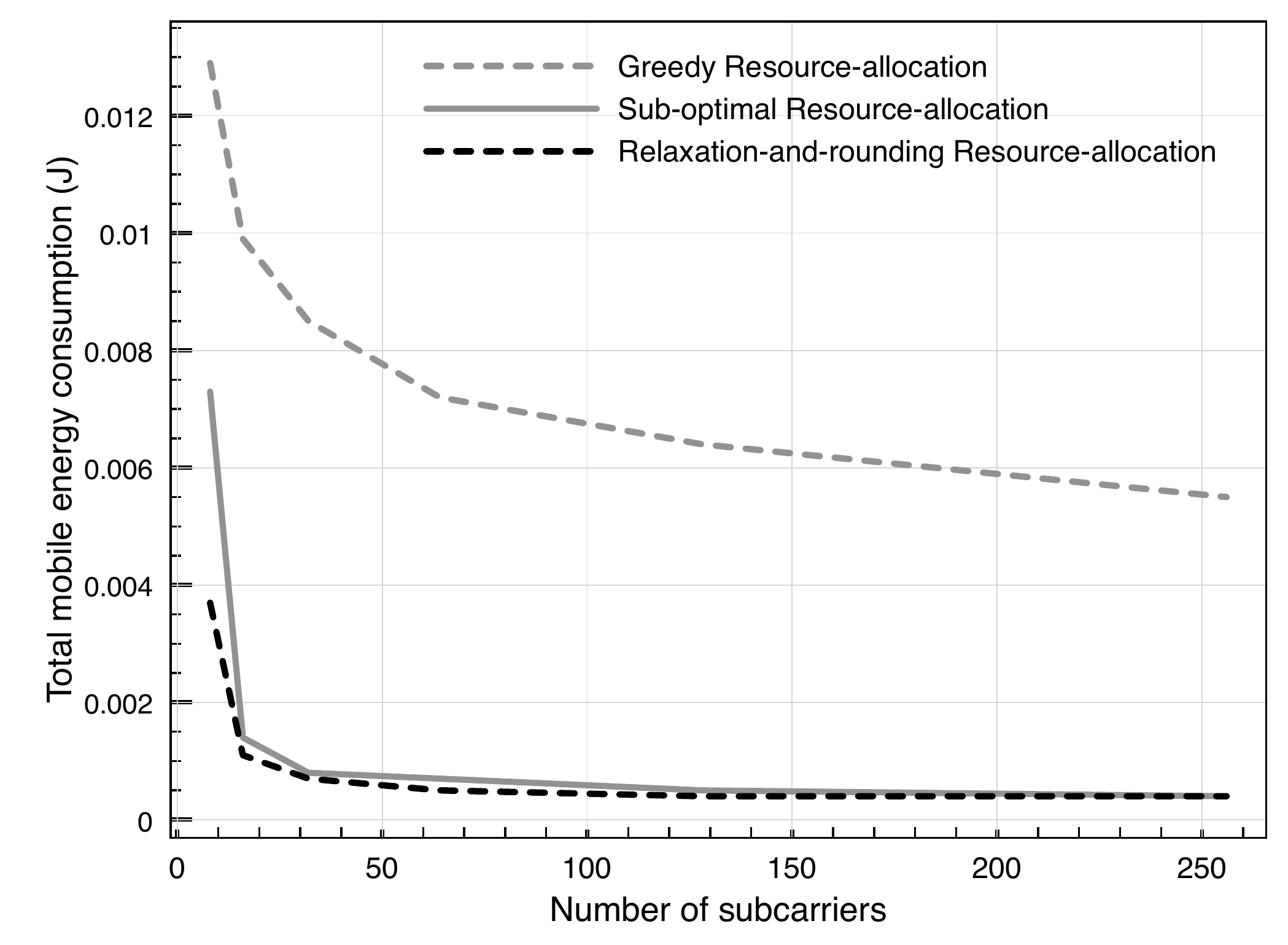}}
\hspace{15mm}
\subfigure[Effect of the number of users. ]{\label{Fig:egy_vs_mobiles_ofdm}
\includegraphics[width=7cm]{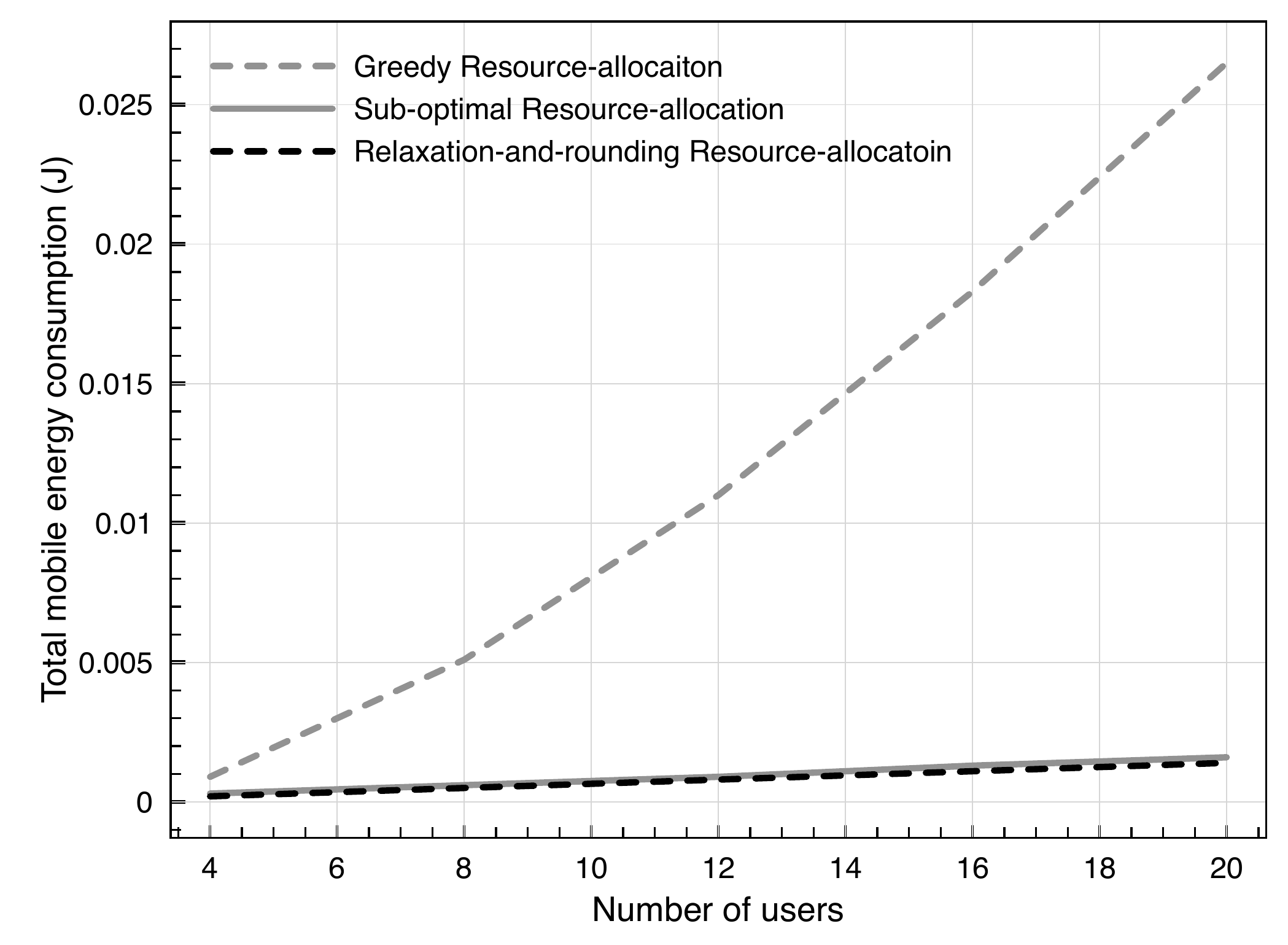}}
\caption{(a) Total mobile energy consumption vs. number of sub-channels for an OFDMA system. (b) Total mobile energy consumption vs. number of users for an OFDMA system.}
\label{Fig:Sim:Static:a}
\end{figure*}

In this section, the performance of the proposed resource-allocation algorithms for both TDMA and OFDMA systems is evaluated by simulation based on $200$ channel realizations. The simulation settings are  as follows unless specified  otherwise.  There are $30$ users  with equal fairness  factors, i.e., $\beta_k=1$ for all $ k$ such that the weighted sum mobile energy consumption represents the total mobile energy consumption. The time slot  $T=100$ ms.  Both channel $h_k$ in TDMA and sub-channel $h_{k,n}$ in OFDAM  are modeled as  independent Rayleigh fading   with average power loss set as $10^{-3}$.  The variance of complex white Gaussian channel noise  $N_0=10^{-9}$ W. Consider mobile $k$. The  computation capacity $F_k$ is  uniformly selected  from the set $\{0.1, 0.2, \cdots, 1.0\}$ GHz  and the local computing energy per cycle $P_k$ follows a uniform distribution in the range $(0, 20\times 10^{-11})$  J/cycle similar to \cite{chen2015efficient}. For the computing task, both the data size and required number of  CPU cycles per bit follow the uniform distribution with $R_k \in [100, 500]$ KB and $C_k \in [500, 1500]$ cycles/bit. All random variables are  independent for different mobiles, modeling heterogeneous mobile computing capabilities. Last, the finite-capacity cloud is modeled by the one with upper-bounded offloaded computation, set as $F=6\times 10^9$ cycles per slot. \footnote{The performance of finite-capacity cloud with non-negligible computing time has similar observations and is omitted due to limited space.}

\subsection{Multiuser MECO for TDMA}

Consider a  MECO system where the bandwidth $B=10$ MHz. For performance comparison, a  baseline \emph{equal resource-allocation} policy is considered, which  allocates equal offloading time duration for mobiles that satisfy $\upsilon_k > 1$ and based on this, the offloaded data sizes are optimized.

The curves of total mobile energy consumption versus the cloud computation capacity are displayed in Fig.~\ref{Fig:egy_vs_capacity}. It can be observed that the performance of the sub-optimal policy  approaches to that of the optimal one when the cloud computation capacity increases and achieves  substantial energy savings gains over the equal resource-allocation policy. Furthermore, the total mobile energy consumption is invariant after the cloud computation capacity exceeds some  threshold (about $6\times 10^9$). This suggests that   there exists some critical value for the cloud computation capacity,  above which   increasing the capacity yields no reduction on the total 
 mobile energy consumption. 

Fig.~\ref{Fig:egy_vs_T} shows the curves of total mobile energy consumption versus the time slot  duration $T$. Several observations can be made. First, the total mobile energy consumption reduces  as the time-slot duration grows. Next,  the sub-optimal policy computed using Algorithm~\ref{Alg:WLA:Sopt} is found to  have close-to-optimal performance and yields total mobile energy consumption  less than half of that for  the equal resource-allocation policy. The energy reduction is more significant for a shorter time slot duration. The reason is that without the optimization on time fractions, the offloading energy of baseline policy grows exponentially as the allocated time fractions decreases.

Next, Fig.~\ref{Fig:egy_vs_mobiles} plots the curves of total energy consumption versus the number of mobiles given fixed cloud computation capacity set as $F=6\times 10^9$ cycles per slot. It shows the total energy consumption of the proposed policy grows with the number of mobiles at a  much slower rate   than that of the equal-allocation policy.  Again, the designed sub-optimal policy is observed to have close-to-optimality.

\subsection{Multiuser MECO for OFDMA}
Consider an OFDMA system where $F=5\times 10^{15}$ cycles per slot (modeling large cloud capacity), $\bar{B}=1$ MHz and $\bar{N}_0=10^{-9}$ W. The proposed low-complexity sub-optimal resource allocation policy is compared with two baseline policies. One is the \emph{relaxation-and-rounding resource-allocation} policy, for which the integer-relaxation convex problem is computed by a convex problem solver, CVX in Matlab, and then the integer solution is determined by rounding technique. The other one is a \emph{greedy resource-allocation} policy. It assigns  each sub-channel  to the user that has highest offloading priority over this sub-channel, followed by the optimal data allocation over assigned sub-channels for each user. However, this policy does not consider the effect of heterogeneous computation loads.

Fig.~\ref{Fig:egy_vs_sub_ofdm} depicts the curves of total mobile energy consumption versus the number of sub-channels in an OFDMA MECO system with $8$ users. It can be observed that the performance of proposed sub-optimal resource allocation is close to that of relaxation-and-rounding policy, especially when the number of sub-channels is large (e.g., $256$). However,  the proposed sub-optimal policy has much smaller computation complexity as discussed in Remark~\ref{Rem:OFDMComplexity}. In addition, the proposed policy has significant energy-savings gain over the greedy policy. The reason is that it considers the varying computation loads over users and allocates more sub-channels to heavy-loaded users, while the greedy policy only offloads computation from users with high priorities. It also suggests that increasing the number of sub-channels has little effect on the total energy savings if this number is above a threshold (about $64$), but otherwise it decreases the total mobile energy consumption significantly. 

Fig.~\ref{Fig:egy_vs_mobiles_ofdm} gives the curves of total mobile energy consumption versus  number of users for an OFDMA system with $128$ sub-channels.  It shows that the energy consumptions for three policies increase with the number of users in the same trend that is almost linear. However, the proposed policy has much smaller increasing rate than the greedy one and approaches the performance of  the relaxation-and-rounding policy.

\section{conclusion}

This work studies resource allocation for a multiuser MECO system based on TDMA/OFDMA, accounting for both the cases of infinite and finite cloud computation capacities. For the TDMA MECO system, it shows that to achieve the minimum weighted sum mobile energy consumption, the optimal resource allocation policy  should have a threshold-based structure. Specifically, we derive an offloading priority function that depends on the local computing energy and channel gains. Based on this, the BS makes a binary offloading decision for each mobile, where users with priorities above and below a given threshold will perform complete and minimum offloading.  
Furthermore, a simple sub-optimal algorithm is proposed to reduce the complexity for computing the threshold for finite-capacity cloud.
 Then, we extend this threshold-based policy structure to the OFDMA system and design a low-complexity algorithm to solve the formulated  mixed-integer optimization problem, which has close-to-optimal performance in simulation.

\appendix
\subsection{Proof of Lemma~\ref{Lem:ConvexProblem}}\label{App:ConvexProblem}
Since $f(x)$ is a convex function, its perspective function \cite{boyd2004convex}, i.e., $t_k f(\frac{\ell_k}{t_k})$, is still convex. Using the same technique in \cite{wang2008power}, jointly considering the cases for $t_k=0$ and $t_k>0$,  $f(x)$ is still convex.
Thus, the objective function, the summation of a set of convex functions, preserves the convexity. Combining it with the linear convex constraints  leads to the result. \!\hfill $\blacksquare$

\subsection{Proof of Lemma~\ref{Lem:PriStrongLA}}\label{App:PriStrongLA}
First, we derive a general result that is the root of equation: $f^{'-1}\!\l(p \r)=g^{-1}\!\l( y \r)$ with respect to $y$ as follows. 
 According to the definitions of $f(x)$ and $g(x)$, it has 
\begin{equation}\label{Eq:Ffunction}
f^{'}(x)=\dfrac{N_0 \ln 2}{B} 2^{\frac{x}{B}} ~\text{and}~ f^{'-1}(y)=B\log_2\l(\dfrac{By}{N_0 \ln 2}\r).
\end{equation} Therefore, the solution for the general equation is
\begin{align}\label{Eq:general}
y&=g(f^{'-1}(p))=f(f^{'\!\!-1}(p))-f^{'\!-\!1}(p) \times f^{'}\!\!(f^{'\!-\!1}(p)) \nn \\
&=f(f^{'-1}(p))- f^{'-1}(p) \times p\nn \\
& =\dfrac{Bp}{\ln 2}-N_0-pB\log_2\l(\dfrac{Bp}{N_0 \ln 2}\r).
\end{align}
Note that to ensure $\ell_k^{*(2)}\ge 0$ in Problem P1, it requires  $f^{'-1}(C_kP_k h_k^2)\ge 0$ from  \eqref{Eq:LandT}. Combining this with \eqref{Eq:Ffunction}, it leads to $v_k \ge 1$ where $v_k$ is defined in \eqref{Eq:Const}. Then, substituting $p=C_k P_k h_k^2$ and $y=\frac{-h_k^2 x}{\beta_k}$ to \eqref{Eq:general} and making arithmetic operations gives the desired result as in \eqref{Eq:OffPriority}.
\hfill $\blacksquare$

\subsection{Proof of Lemma~\ref{Lem:Phi}}\label{App:Phi}
First, the monotone increasing property in terms of $\beta$ is straightforward, since the offloading priority function in \eqref{Eq:OffPriority} is linear to $\beta$. Next, by rewriting \eqref{Eq:OffPriority} as 
$$\phi(\beta, C, P, h)=\beta BC P \l[ \log_2\l(\frac{B C P h^2}{N_0 \ln 2}\r)  -\dfrac{1}{\ln2}\r]+\dfrac{\beta N_0}{h^2},$$ it is easy to conclude that $\phi(\beta, C, P, h)$ is monotone increasing with respect to $C$ and $P$. Last, the first derivative of $\phi(\beta, C, P, h)$ for $h$ can be derived as: 
\begin{align*}
\dfrac{\partial \phi(\beta, C, P, h)}{\partial h}=\frac{2\beta(BCPh^2-N_0 \ln2)}{h^3 \ln 2}.
\end{align*}
For $\upsilon=\frac{BCPh^2}{N_0 \ln2}\ge 1$, we have $\dfrac{\partial \phi(\beta, C, P, h)}{\partial h}\ge0$, leading to the desired results.
\hfill $\blacksquare$

\subsection{Proof of Theorem~\ref{Theo:OptiPolicyP2}}\label{App:OptiPolicyP2}
First, to prove this theorem, we need the following two lemmas which can be easily proved using the definition of Lambert function and its property.
\begin{lemma}\label{Lem:GInverse}\emph{The function $g^{-1}(y)$ can be expressed as $g^{-1}(y)=\dfrac{B \l[W_0(\frac{y+N_0}{-N_0 e})+1\r]}{\ln 2}.$
}
\end{lemma}

\begin{lemma}\label{Lem:IncreaseG}\emph{
The function $g^{-1}(y)$ is a monotone decreasing function for $y<0$.}
\end{lemma}

Then, consider case 1) in Theorem~\ref{Theo:OptiPolicyP2}. Note that  for mobile $k$, if  $m_k^+=0$ and $\upsilon_k\le 1$, it results in $\ell_k^{*(2)}=0$ derived from \eqref{Eq:OptL}. Thus, if these two conditions are satisfied for all $k$, it leads to $\ell_k^{*(2)}=t_k^{*(2)}=0$. 

For case 2), if there exists  mobile $k$ such that $\upsilon_k >1$ or $m_k^+>0$, it leads to $\ell_k^{*(2)}>0$. And the time-sharing constraint should be active since remaining time can be used for extending offloading duration so as to reduce  transmission energy. 
Moreover, consider each user, say user $k$. If $\upsilon_k \ge 1$, then from \eqref{Eq:OptL} and \eqref{Eq:OptT}, $\{\ell_k^{*(2)}, t_k^{*(2)}\}$ should satisfy the following:
\begin{subequations}
\begin{align}
\dfrac{\ell_k^{*(2)}}{ t_k^{*(2)}}&\!=\!\min\l\{~ \max \l[\dfrac{m_k^+}{t_k^{*(2)}}, f^{'-1}\!\l(C_k P_k h_k^2 \r)\r], ~\dfrac{R_k}{t_k^{*(2)}}\r\} \label{Eq:TL1} \\
&\!=\!\max \l\{~ \dfrac{m_k^+}{t_k^{*(2)}}, ~\min \l[f^{'-1}\!\l(C_k P_k h_k^2 \r), \dfrac{R_k}{t_k^{*(2)}}\r]\r\} \label{Eq:TL2} \\
&=g^{-1}\!\l(\frac{-h_k^2 \lambda^*}{\beta_k} \r). \label{Eq:TL3}
\end{align}
\end{subequations}
Using Lemma~\ref{Lem:PriStrongLA} and Lemma~\ref{Lem:IncreaseG}, we have the following:
\begin{enumerate}
\item{If $\phi_k>\lambda^*\ge 0$, it has $-h_k^2 \phi_k<-h_k^2 \lambda^*\le 0$.  Then,  from \eqref{Eq:TL1}, it gives 
 \begin{align}\label{Eq:HighPri}
 \max & \l[\dfrac{m_k^+}{t_k^{*(2)}}, f^{'-1}\!\l(C_k P_k h_k^2 \r)\r] \ge f^{'-1}\!\l(C_k P_k h_k^2\r) \nn \\ 
 &=g^{-1}\!\l(\frac{-h_k^2 \phi_k}{\beta_k} \r)>g^{-1}\!\l(\frac{-h_k^2 \lambda^*}{\beta_k} \r). 
 \end{align} 
 From \eqref{Eq:TL1}, \eqref{Eq:TL3} and \eqref{Eq:HighPri}, it follows that
$\ell_k^{*(2)}=R_k$.}
\item{If $\phi_k=\lambda^*$, it has $f^{'-1}\!\l(C_k P_k h_k^2\r)=g^{-1}\!\l(\frac{-h_k^2 \lambda^*}{\beta_k} \r)$.} 
\item{If $0\le \phi_k<\lambda^*$, it has $-h_k^2 \phi_k>-h_k^2 \lambda^*$. Combining it with \eqref{Eq:TL2} leads to 
 \begin{align}\label{Eq:LowPri}
 \min &\l[ f^{'-1}\!\l(C_k P_k h_k^2 \r),\dfrac{R_k}{t_k^{*(2)}},\r] \le f^{'-1}\!\l(C_k P_k h_k^2\r) \nn \\
 &=g^{-1}\!\l(\frac{-h_k^2 \phi_k}{\beta_k} \r)<g^{-1}\!\l(\frac{-h_k^2 \lambda^*}{\beta_k} \r).
 \end{align}
From \eqref{Eq:TL2}, \eqref{Eq:TL3} and \eqref{Eq:LowPri}, it follows that
$\ell_k^{*(2)}=m_k^+$.}
\end{enumerate}
Furthermore, if $\upsilon_k\!<\!1$, it has $\ell_k^{*(2)}\!=\!m_k^+$. Note that this case can be included in the scenario of $\phi_k\!<\!\lambda^*$ with the definition of $\phi_k$ in \eqref{Eq:OffPriority}. Last, from \eqref{Eq:TL3}, it follows that 
\begin{equation} \label{Eq:time}
t_k^{*(2)}=\dfrac{\ell_k^{*(2)}}{g^{-1}\!\l(\frac{-h_k^2 \lambda^*}{\beta_k} \r)}=\dfrac{\ell_k^{*(2)} \ln 2}{B\l[W_0(\frac{\lambda^* h_k^2/\beta_k-N_0}{N_0 e})+1\r]}
\end{equation} where \eqref{Eq:time} is obtained using Lemma~\ref{Lem:GInverse}, ending the proof. \hfill $\blacksquare$
\subsection{Proof of Lemma~\ref{Lem:OffPrioNLA}}\label{App:OffPrioNLA} 
First, by arithmetic operations with the Lambert function, it can be proved that  the solution for a general equation $x \ln x+ p x=q$ is $x=\dfrac{q}{W_0(q \times e^p )}.$ 

Next,  to solve equation \eqref{Eq:FunPriCloud2}, let $y_k=C_k P_k h_k^2-\dfrac{x C_k h_k^2}{\beta_k F^{'} }$ and use the derivation method in Lemma~\ref{Lem:PriStrongLA}, it has
\begin{equation}\label{Eq:y}
\dfrac{F^{'} y_k}{C_k}-F^{'} P_k h_k^2=\dfrac{By_k}{\ln 2}-N_0-y_kB\log_2\l(\dfrac{By_k}{N_0 \ln 2}\r).
\end{equation}
Defining $z_k=\dfrac{By_k}{N_0 \ln 2}$, \eqref{Eq:y} can be rewritten as
\begin{equation}\label{Eq:z}
z_k \ln z_k +(a_k-1)z_k=b_k-1,
\end{equation}
where $a_k$ and $b_k$ are defined in Lemma~\ref{Lem:OffPrioNLA}.
Using Lambert function, the solution for \eqref{Eq:z} can be obtained: $z_k=\widehat{\upsilon}_k$ where $\widehat{\upsilon}_k$ is defined in \eqref{Eq:Upsilon3}.
Then, it follows  that
\begin{align}
x&\overset{(a)}{=} \beta_k F^{'} P_k(1-\dfrac{N_0 \ln 2}{B P_k C_k h_k^2} z_k) \nn \\&\overset{(b)}{=} \beta_k F^{'} P_k(1-\dfrac{a_k}{b_k} z_k) \nn \\&\overset{(c)}{=}\dfrac{\beta_k N_0}{h_k^2} (z_k \ln z_k -z_k+1) \label{Eq:x}
\end{align}
where $(a)$ comes from the relationship among $x, y_k$ and $z_k$; $(b)$ follows the definition of $a_k$ and $b_k$; $(c)$ is derived from \eqref{Eq:z}.
This leads to  the desired result.\hfill $\blacksquare$

\subsection{Proof of Lemma~\ref{Lem:WandAB}}\label{App:WandAB} 
It is equivalent to proved as below that when $\widehat{\upsilon}>1$, it has $b \ge a$.
According to the definition of Lambert function, it has $b-1= W_0((b-1)e^{(b-1)}).$
Then, it leads to \begin{align}\label{Eq:WandAB}
\widehat{\upsilon}\!=\!\dfrac{b-1}{ W_0((b-1)e^{(a-1)})}=\dfrac{W_0((b-1)e^{(b-1)})}{ W_0((b-1)e^{(a-1)})} \ge 1.
\end{align}

Using the monotone increasing property of Lambert function, \eqref{Eq:WandAB} is equivalent to  $b \ge a$. \hfill $\blacksquare$


\begin{IEEEbiography}
[{\includegraphics[width=1in,clip,keepaspectratio]{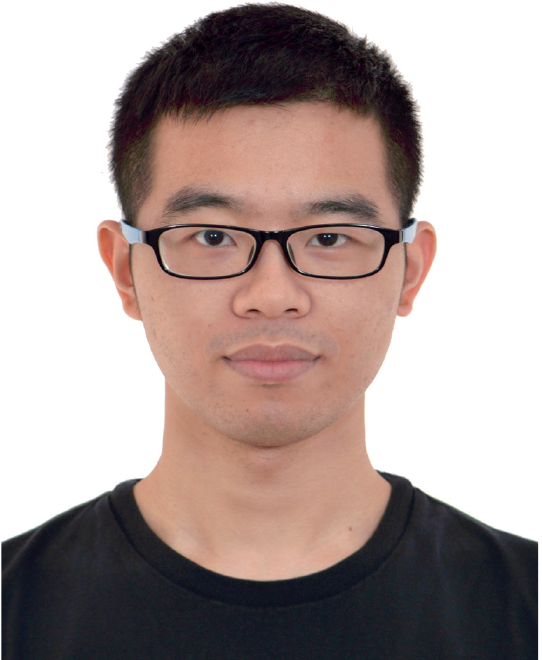}}]
{Changsheng You}(S'15) received the B.S. degree in electronic engineering and information science from the University of Science and Technology of China (USTC) in 2014. He is currently working towards the Ph.D. degree in electrical and electronic engineering at The University of Hong Kong (HKU). His research interests include mobile-edge computing, fog computing, wireless power transfer, energy harvesting systems and convex optimization.
\end{IEEEbiography}

\begin{IEEEbiography}
[{\includegraphics[width=1in,clip,keepaspectratio]{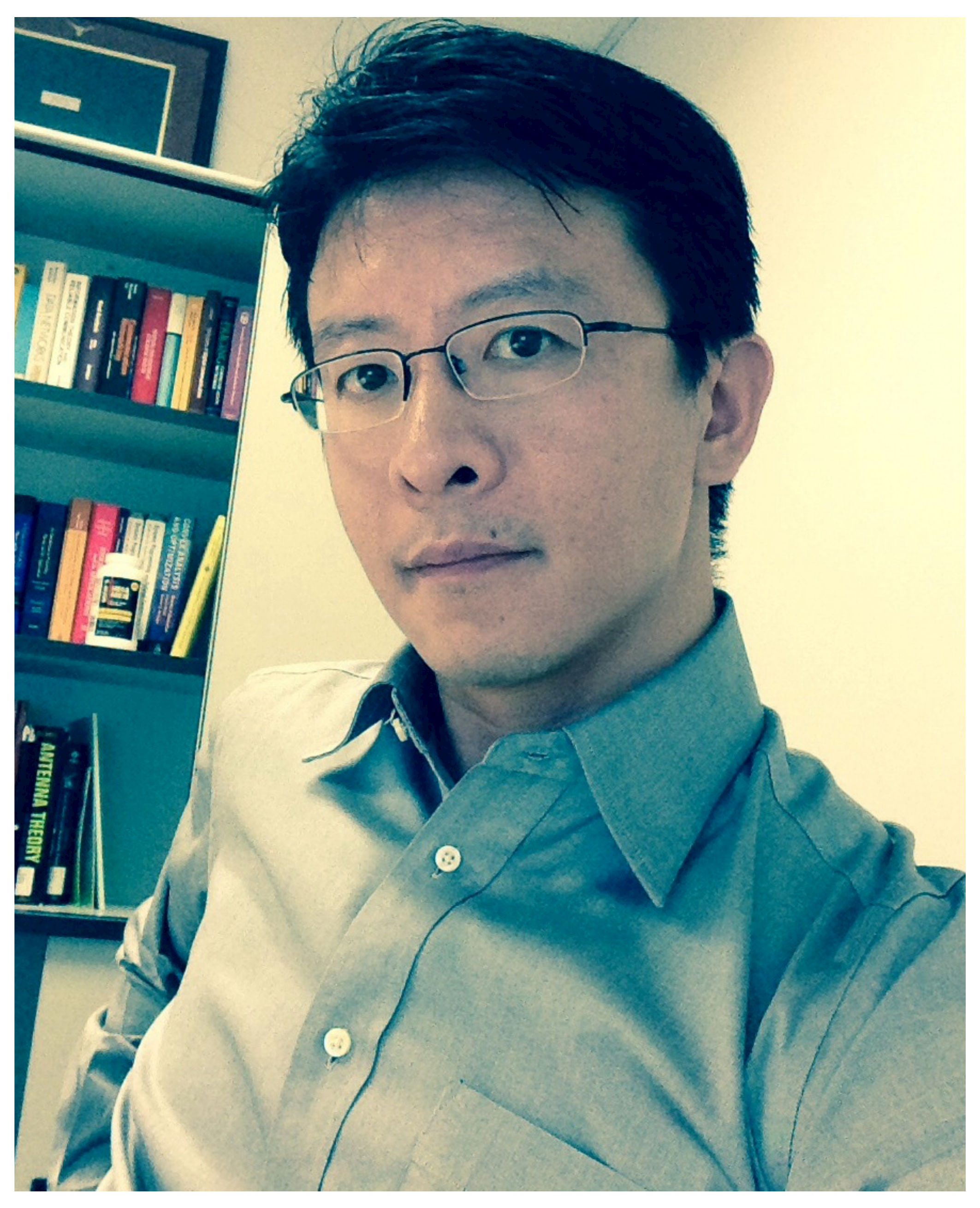}}]
{Kaibin Huang} (M'08-SM'13) received the B.Eng. (first-class hons.) and the M.Eng. from the National University of Singapore, respectively, and the Ph.D. degree from The University of Texas at Austin (UT Austin), all in electrical engineering.

Since Jan. 2014, he has been an assistant professor in the Dept. of Electrical and Electronic Engineering (EEE) at The University of Hong Kong. He is an adjunct professor in the School of EEE at Yonsei University in S. Korea. He used to be a faculty member in the Dept. of Applied Mathematics (AMA) at the Hong Kong Polytechnic University (PolyU) and the Dept. of EEE at Yonsei University. He had been a Postdoctoral Research Fellow in the Department of Electrical and Computer Engineering at the Hong Kong University of Science and Technology from Jun. 2008 to Feb. 2009 and an Associate Scientist at the Institute for Infocomm Research in Singapore from Nov. 1999 to Jul. 2004. His research interests focus on the analysis and design of wireless networks using stochastic geometry and multi-antenna techniques.

He frequently serves on the technical program committees of major IEEE conferences in wireless communications. He has been the technical chair/co-chair for the IEEE CTW 2013, the Comm. Theory Symp. of IEEE GLOBECOM 2014, and the Adv. Topics in Wireless Comm. Symp. of IEEE/CIC ICCC 2014 and has been the track chair/co-chair for IEEE PIMRC 2015, IEE VTC Spring 2013, Asilomar 2011 and IEEE WCNC 2011. Currently, he is an editor for IEEE Journal on Selected Areas in Communications (JSAC) series on Green Communications and Networking, IEEE Transactions on Wireless Communications, IEEE Wireless Communications Letters. He was also a guest editor for the JSAC special issues on communications powered by energy harvesting and an editor for IEEE/KICS Journal of Communication and Networks (2009-2015). He is an elected member of the SPCOM Technical Committee of the IEEE Signal Processing Society. Dr. Huang received the 2015 IEEE ComSoc Asia Pacific Outstanding Paper Award, Outstanding Teaching Award from Yonsei, Motorola Partnerships in Research Grant, the University Continuing Fellowship from UT Austin, and a Best Paper Award from IEEE GLOBECOM 2006 and PolyU AMA in 2013. 
\end{IEEEbiography}

\begin{IEEEbiography}
[{\includegraphics[width=1in,clip,keepaspectratio]{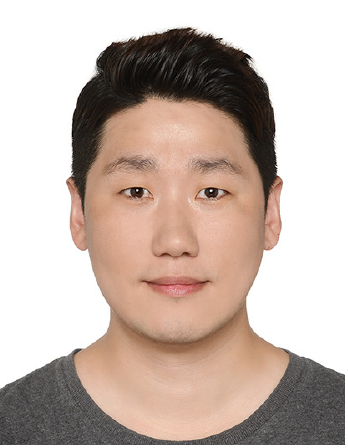}}]
{Hyukjin Chae} received the B.S and Ph.D degree in electrical and electronic
engineering from Yonsei University, Seoul, Korea. He joined LG Electronics,
Korea, as a Senior Research Engineer in 2012. His research interests
include interference channels, multiuser MIMO, D2D, V2X, and full duplex
radio. From Sep. 2012, he has contributed and participated as a delegate in
3GPP RAN1 with interests in ePDCCH, eIMTA, FD MIMO, Indoor positioning,
D2D, and V2X communications. He is an inventor of more than 100 patents.
\end{IEEEbiography}

\begin{IEEEbiography}
[{\includegraphics[width=1in,clip,keepaspectratio]{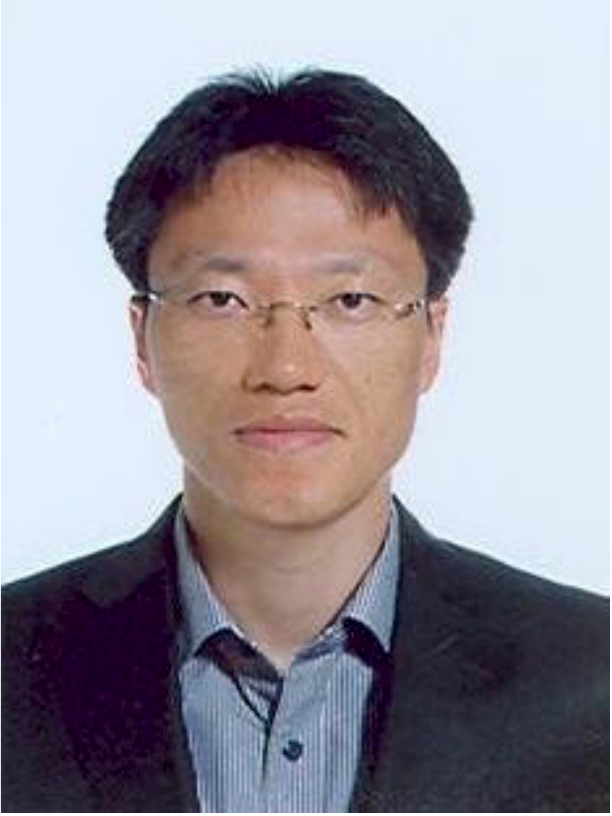}}]
{Byoung-Hoon Kim}  (S'95-M'01) received the B.S. and M.E. degrees in electronics engineering, and the Ph.D. degree in electrical engineering and computer science from Seoul National University, Seoul, Korea, in 1994, 1996, and 2000, respectively. From 2000 to 2003, he was with GCT Semiconductor, Seoul, developing W-CDMA and WLAN chip sets. From 2003 to 2008, he was with QUALCOMM Incorporated, San Diego, CA, where he was responsible for MIMO technology development and 3GPP LTE standard and design works. Since 2008, he has been with LG Electronics as Vice President and Research Fellow, developing advanced wireless technologies including 5G mobile communications and 3GPP LTE-Advanced/5G standards. He was also involved in IEEE 802.11 standard works and assumed the role of a member of board of directors of Wi-Fi Alliance from 2011 to 2012. His current research interest includes advanced channel coding, multiple access, V2X, massive MIMO, flexible and full duplex radio, and mmWave technologies.
Dr. Kim is co-author of Scrambling Techniques for CDMA Communications (Springer, 2001) and was elected as the 1st IEEE Communications Society Asia-Pacific Best Young Researcher in 2001.
\end{IEEEbiography}

\end{document}